\documentclass[11pt,letterpaper]{article}

\usepackage{url}
\usepackage{fullpage}

\usepackage[subtle,mathspacing=normal]{savetrees}
\usepackage{amsmath}
\usepackage{algorithm}
\usepackage[noend]{algpseudocode}
\usepackage{amssymb}
\usepackage{amsthm}
\usepackage{color}
\usepackage{array}
\usepackage{xy}
\usepackage{setspace}
\usepackage{varwidth}
\usepackage{multicol}
\usepackage{algorithm}
\usepackage{hyperref}
\usepackage{todonotes}
\renewcommand{\todo}[1]{}

\newcommand{\poly}{\operatorname{poly}}

\newcommand{\polylog}{\operatorname{polylog}}

 \newcommand{\E}{\mathbb{E}}

\newcommand{\free}{\operatorname{free}}
\newcommand{\maxx}{\widehat{\max}}

\newtheorem{thm}{Theorem}[section]

\theoremstyle{remark}

\newtheorem{theorem}{Theorem}[section]

\newtheorem{lemma}[thm]{Lemma}
\newtheorem{proposition}[thm]{Proposition}

\newtheorem{corollary}[thm]{Corollary}

\theoremstyle{remark}

\newcommand{\defn}[1]{\textbf{\emph{#1}}}
\renewcommand{\paragraph}[1]{\vspace{.2 cm} \noindent \textbf{#1}}

\newcommand{\coord}{\text{coord}}
\newcommand{\trans}{\text{trans}}
\renewcommand{\frame}{\text{frame}}
\newcommand{\surplus}{\text{surplus}}
\newcommand{\round}{\text{round}}
\newcommand{\gain}{\text{gain}}
\newcommand{\slope}{\text{slope}}

\newcommand{\Pois}{\operatorname{Pois}}

\usepackage{thmtools} 
\usepackage{thm-restate}

\usepackage{authblk}

\begin{document}

\title{Tight Analyses of Ordered and Unordered Linear Probing}
\author{Mark Braverman\footnote{Princeton, \texttt{mbraverm@cs.princeton.edu }} \phantom{f}and William Kuszmaul\footnote{CMU, \texttt{kuszmaul@cmu.edu}}}
\date{}

\maketitle

\begin{abstract}
   Linear-probing hash tables have been classically believed to support insertions in time $\Theta(x^2)$, where $1 - 1/x$ is the load factor of the hash table. Recent work by Bender, Kuszmaul, and Kuszmaul (FOCS'21), however, has added a new twist to this story: in some versions of linear probing, if the \emph{maximum} load factor is at most $1 - 1/x$, then the \emph{amortized} expected time per insertion will never exceed $x \polylog x$ (even in workloads that operate continuously at a load factor of $1 - 1/x$).  Determining the exact asymptotic value for the amortized insertion time remains open.

   In this paper, we settle the amortized complexity with matching upper and lower bounds of $\Theta(x \log^{1.5} x)$. Along the way, we also obtain tight bounds for the so-called path surplus problem, a problem in combinatorial geometry that has been shown to be closely related to linear probing. We also show how to extend Bender et al.'s bounds to say something not just about \emph{ordered} linear probing (the version they study) but also about classical linear probing, in the form that is most widely implemented in practice.
\end{abstract}
\thispagestyle{empty}
\newpage

\clearpage
\setcounter{page}{1}

\section{Introduction}

The linear probing hash table is one of the oldest and most widely used data structures in computer science \cite{Sedgewick83,Sedgewick90, Kruse84, LewisDe91, Weiss00, MainSa01, Standish95, DrozdekSi95, TremblaySo84}. The classical perspective on linear probing is that it suffers irredeemably from \emph{clustering effects} that cause it to perform badly at high load factors. Recent work by Bender et al.~\cite{bender2022linear}, however, has shown that the reality is more subtle: in some versions of linear probing, there is actually a hidden \emph{anti-clustering} effect that reduces (amortized) insertion cost far below what was previously thought to be the case.\footnote{This result is for a variation known as \emph{ordered} linear probing. However,  as we shall see later, with some additional ideas, one can also prove an analogous result even for the completely classical version of the data structure.}

Bender et al.'s analysis takes a data structure that for more than half a century was thought to be fully understood, and cracks it wide open once again: they prove that the true amortized expected insertion cost, at load factor $1 - 1/x$, is of the form not $\Theta(x^2)$ (as classically thought) but instead $\Theta(x \cdot f(x))$ for some $f(x)$ between $\poly \log \log x$ and $\poly \log x$. Determining the asymptotic value of $f(x)$ remains open.

In this paper, we establish a tight (and somewhat unexpected) bound of $f(x) = \Theta( \log^{1.5} x)$. Along the way, we demonstrate that the timing characteristics of linear probing (one of the simplest data structures in computer science!) are actually determined by a remarkably intricate (and beautiful) underlying combinatorial structure.

\paragraph{Background on linear probing.} In its most basic form, the linear-probing hash table can be described as follows. At any given moment, the hash table is an array consisting of elements, free slots, and tombstones (i.e., elements marked as deleted). We can insert an item $u$ by computing its hash $h(u) \in [n]$, examining array positions $h(u), h(u) + 1, h(u) + 2, \ldots$ (modulo the array size $n$), and overwriting the first free slot or tombstone that we encounter with $u$. We can query $u$ by examining slots $h(u), h(u) + 1, h(u) + 2, \ldots$ until we either find $u$ or until we encounter a free slot (at which point we conclude that $u$ is not present). And we can delete $u$ by simply replacing it with a tombstone. Finally, to prevent the over-accumulation of tombstones, one should periodically perform rebuilds in which all of the tombstones are cleared out.

A major appeal of linear probing is its data locality: each operation accesses just one small contiguous region of memory. In practice, linear probing hash tables are a go-to choice for settings where cache misses and memory bandwidth are a priority (see, e.g., discussion in \cite{RichterAlDi15}).

However, linear probing also comes with a major drawback: as the hash table fills up, the elements cluster together into \defn{runs} that are longer than one might intuitively expect. If the hash table is filled to $1 - 1 /x$ full, then the expected length of the run containing a given element becomes $\Theta(x^2)$, and the expected time per insertion is therefore also $\Theta(x^2)$ \cite{Knuth63, Knuth98Vol3}. This phenomenon, which was first discovered by Donald Knuth in 1963 \cite{Knuth63} (as well as by Konheim and Weiss in 1966~\cite{KonheimWe66}), is often referred to as \defn{primary clustering} \cite{Sedgewick83,Sedgewick90, Kruse84, LewisDe91, Weiss00, MainSa01, Standish95, DrozdekSi95, TremblaySo84}.

It is worth also mentioning a second phenomenon that Knuth discovered in the same 1963 paper \cite{Knuth63}, namely, an asymmetry between insertions and queries. Intuitively, $\Theta(x^2)$ expected-time insertions would seem to imply $\Theta(x^2)$ expected-time queries (and, indeed, they do, as one can simply query the most recently inserted element). However, if a query targets a \emph{random} element out of those present, then the expected query time drops to a much more reasonable bound of $\Theta(x)$. In practice, this means that clustering is primarily a problem for \emph{insertions}, not queries. 

In theoretical discussions of linear probing, it is often helpful to formalize this asymmetry between queries and insertions by adding the following optimization: modify the hash table so that, within each run, the elements are always stored in sorted order by hash. The advantage of this modification is that, now, it is not just the \emph{average-case} expected query time that is $O(x)$, but actually the \emph{worst-case} expected query time as well. This version of the hash table is known as \defn{ordered linear probing} \cite{AmbleKn74} (or, in some parts of the literature, as \defn{robin-hood hashing} \cite{CelisLaMu85}). In keeping with previous work \cite{bender2022linear}, we will focus most of our discussion on ordered linear probing; however, as we shall see, all of the results in this paper also have natural analogues for a classical linear-probing hash table.


\paragraph{A recent twist: anti-clustering.} Consider a sequence of insertions/deletions that keeps the hash table at or below $1 - 1/x$ full at all times. The main result in \cite{bender2022linear} is that, even though \emph{some} insertions take $\Theta(x^2)$ expected time, the \emph{amortized expected time} per insertion never exceeds $O(x \polylog x)$ \cite{bender2022linear}. 


\begin{theorem}[Upper bound of \cite{bender2022linear}]
Consider an ordered linear probing hash table, where deletions are implemented with tombstones, and where rebuilds are performed every $n / \polylog x$ insertions/deletions. If the load factor of the hash table stays at or below $1 - 1/x$ at all times, then the amortized expected time per insertion/deletion is $O(x \polylog x)$ and the expected time per query is $O(x)$.
\label{thm:prev}
\end{theorem}

To fully appreciate Theorem \ref{thm:prev}, it is helpful to consider the scenario in which the user first fills the hash table to $1 - 1/x$ full (call this time window $[t_0, t_1]$), and then alternates between insertions and deletions (call this time window $(t_1, t_2]$). Even though the insertion at time $t_1$ takes expected time $\Theta(x^2)$, most of the insertions in $[t_0, t_1]$ were at much lower load factors, so the amortized expected time per insertion in $(t_0, t_1]$ is actually $O(x)$. What is surprising is that, even as we continue into time window $(t_1, t_2]$, where all insertions are at load factor $1 - 1/x$, the amortized expected time per insertion remains $\tilde{O}(x)$. 

The good performance in $(t_1, t_2]$ is due to an \emph{anti-clustering} phenomenon, in which the tombstones created by deletions break runs apart \emph{more efficiently} than the insertions are able to connect new runs together.\footnote{This anti-clustering phenomenon appears to have gone unobserved in both the theory and practical literatures up until the work of \cite{bender2022linear}. Historically, tombstones were viewed as somewhat of an afterthought in the design of linear probing---one of several equally good ways that deletions could be implemented (see historical discussion in \cite{bender2022linear}). It is worth noting, however, that the most widely-used high-performance hash tables in practice (e.g., the hash tables released by Google \cite{Abseil17} and Facebook \cite{BronsonSh19}) do, in fact, use tombstones, indicating that the anti-clustering effect may have implicitly influenced their designs.} This anti-clustering effect also explains why the rebuild window in Theorem \ref{thm:prev} is set to be relatively large (at $R = n / \polylog x$ rather than, say, $R = \Theta(n / x)$). This distinction gives tombstones more time to perform anticlustering before being removed. The optimal choice for $R$ remains open, but it is known to be between $n / \polylog x$ and $n$. 


But what about the $\polylog x$ term in the insertion time? Is it real, or is it an artifact of the analysis? This question remains largely open. The authors of \cite{bender2022linearfull} are able to show, however, that the true insertion time is \emph{some} super-linear function in $x$:

\begin{theorem}[Lower bound of \cite{bender2022linear}]
Consider an ordered linear probing hash table, where deletions are implemented with tombstones, and where rebuilds are performed every $R$ insertions/deletions for some $R$. Consider a workload that fills the hash table to $1 -1/x$ full and then alternates between deleting a random item (out of those present), inserting a new item (never before inserted), and querying a random item (out of those present). No matter the choice of $R$, the amortized expected time per operation (including queries) will be at least $x(\log \log x)^{\Omega(1)}$. 
\label{thm:prev2}
\end{theorem}

Combined, the upper and lower bounds of \cite{bender2022linear} put us in an interesting situation. Once again, the analysis of (ordered) linear probing hash tables, arguably one of the most basic data structures in computer science, stands as an open question. What is the asymptotic amortized expected cost per insertion/deletion? Is it something like $x \log \log x$ (which, in practice, might be as good as $O(x)$) or $x \log^3 x$ (which, in practice, might be no better than $x^2$)? And, at an algorithmic level, what is the optimal choice of rebuild-window size $R$?

\paragraph{The combinatorial bottleneck: the path surplus problem.} Interestingly, the question of analyzing ordered linear probing actually reduces to a remarkably simple (and, at first glance, seemingly unrelated) problem from combinatorial geometry.

Consider a grid $[0, m] \times [0, m]$ containing $B \sim \Pois(m^2)$ blue dots and $R \sim \Pois(m^2)$ red dots placed uniformly at random. For each monotonic path through the grid, going from the bottom left to the top right, look at the dots \emph{underneath} the path, and define path's \defn{surplus} to be the number of blue dots minus the number of red dots (see, e.g., Figure \ref{fig:surplus}). Then the quantity we care about is the expected value of the \emph{maximum} surplus over all monotonic paths.

This \defn{path surplus problem}, it turns out, is what dictates the behavior of (ordered) linear probing hash tables at high load factors \cite{bender2022linear}. Indeed, the reason that Theorems \ref{thm:prev} and \ref{thm:prev2} hold is because the expected maximum path surplus turns out to be between $m (\log \log m)^{\Omega(1)}$ and $m (\log m)^{O(1)}$.

The main technical bottleneck to resolving the insertion time of ordered linear probing is to obtain tight bounds on the path surplus problem. Note that it is not clear, \emph{a priori}, what these bounds should be. The known arguments \cite{bender2022linear}, although both simple and elegant (we will summarize them in Section \ref{sec:prelims}), do not seem to offer any hint at the final answer. Any solution that establishes tight bounds would appear to require a significantly new approach.

\paragraph{This paper: tight bounds on path surplus and insertion time.}
The main technical contribution of this paper is a new (and completely different) analysis of the path surplus problem. Our analysis gives a surprising answer to Bender et al.'s question: the expected maximum surplus over all monotonic paths is $\Theta(m \log^{0.75} m)$. 

Our second contribution is a tighter analysis of the connection between path surpluses and ordered linear probing. From this, we are able to recover exact asymptotic bounds for ordered linear-probing. If the rebuild-window length $R$ is chosen optimally, then worst-case amortized expected time per insertion/deletion becomes
$$\Theta(x \log^{1.5} x),$$
while the expected query time remains $O(x)$. More generally, the relationship between the three quantities can be captured with the following theorem:

\begin{restatable}{theorem}{mainthm}
Let $x$ and $n$ be parameters satisfying $x \le n^{o(1)}$. Consider an ordered linear probing hash table on $n$ slots that performs rebuilds every $R = n / \beta$ insertions, where $1 \le \beta \le x$, and where deletions are implemented with tombstones. Under the condition that the load factor never exceeds $1 - 1/x$, the worst-case amortized expected insertion/deletion time is
$$\Theta\left(x \log^{1.5} x  + \beta x\right)$$
and the worst-case expected query time is
$$\Theta\left(x + \frac{x \log^{1.5} x}{\beta}\right).$$ 
\label{thm:main}
\end{restatable}
The trade-off curve reveals $R = \Theta(n / \log^{1.5} x)$ as the unique optimal value for $R$. When $R \ll x / \log^{1.5} x$, insertions become needlessly slow, and when $R \gg x / \log^{1.5} x$, queries become needlessly slow. Only at $R = \Theta(n / \log^{1.5} x)$ do both operations achieve their optimal running times.

\paragraph{Completing the circle: what about classical linear probing?} The ultimate goal of studying linear-probing hash tables is to understand how we should think about the basic (typically unordered) data structure that is commonly used in practice \cite{RichterAlDi15, WikipediaLinearProbing}.\footnote{This is not to say that ordered linear probing is \emph{never} used. It is, for example, widely used to implement filters \cite{BenderFaJo12, PandeyCoDu21, GeilFO18}.}

Recall that the original purpose of the \emph{ordering} optimization \cite{AmbleKn74} was to transform Knuth's \cite{Knuth63} \emph{average-case} analysis of queries into a \emph{worst-case} analysis. This raises the following question: is there an \emph{average-case} version of Theorem \ref{thm:main} that holds for classical unordered linear probing?

Our final result answers this question in the affirmative:

\begin{restatable}{theorem}{thmunordered}
    Let $n, x, \beta$ be parameters such that $x = n^{o(1)}$, such that $c \le \beta \le x$ for some sufficiently large positive constant $c$. Consider an \emph{unordered} linear-probing hash table, implemented using tombstones and with rebuild-window size $R = n / \beta$, and subjected to the following \defn{average-case workload}: The hash table is filled to $1 - 1/x$ full, and then alternates between deleting a random element out of those present and inserting a new (never-before-inserted) element.
    
    Then the amortized expected time per insertion/deletion is 
    $\Theta(x \log^{1.5} x + \beta^{-1} x),$ 
    and the worst-case expected cost of querying a \emph{random} element out of those present is
    $\Theta(x + \beta^{-1} x \log^{1.5} x).$
    \label{thm:unordered}
\end{restatable}

\begin{restatable}{corollary}{corunordered}
    In the context of Theorem \ref{thm:unordered}, the optimal rebuild-window size $R$ is any $R  = \Theta(n / \log^{1.5} x)$, at which point the amortized expected insertion/deletion time is $O(x \log^{1.5} x)$ and the expected query time is $O(x)$.
    \label{cor:unordered}
\end{restatable}

Theorem \ref{thm:unordered} suggests that anti-clustering should be viewed as a phenomenon that takes place not just in ordered linear-probing hash tables, but in all linear-probing hash tables. It offers at least a partial explanation for why linear probing, despite concerns about clustering, remains so widely used in practice.

It should be noted that both the restriction in Theorem \ref{thm:unordered}---that queries/deletions are to \emph{random} elements and that insertions are to \emph{new} elements---are fundamentally necessary for the theorem to hold. If either restriction is relaxed, then one can force $\Theta(x^2)$-time operations by either repeatedly inserting/deleting the same element over and over, or repeatedly querying the first element to be inserted during the rebuild window. 

\paragraph{Paper outline. }The rest of the paper is structured as follows. We begin in Section \ref{sec:prelims} with preliminaries, including a summary of past techniques. Then, in Section \ref{sec:pathsurplus}, we achieve our main technical result, which is a bound of $\Theta(m \log^{0.75} m)$ on path surplus. After this, in Sections \ref{sec:insertionsurplus} and \ref{sec:linprobing}, we tighten the relationship between the path surplus problem and ordered linear-probing hash tables, so that we can transform our tight bounds for the former into tight bounds for the latter. Finally, Section \ref{sec:unordered}, we extend our results to unordered linear probing---as we shall see, this extension itself also requires several new technical ideas.

\section{Preliminaries and Background}\label{sec:prelims}

\paragraph{Ordered linear probing.}
An ordered linear probing hash table with \defn{capacity} $ n $ can be viewed as an array of size $ n $ consisting of free slots, elements (also known as \defn{items} or \defn{keys}), and tombstones (i.e., elements marked as deleted). The hash table also has access to a fully random hash function $h$ mapping elements to random indices $\{1, 2, \ldots, n\}$.\footnote{In general, one does not necessarily need to assume that $ h $ is fully random, as there well-understood techniques for simulating fully random hash functions in this setting (see, e.g., \cite{siegel2004universal, pagh2008uniform, DietzfelbingerWo03, bercea2023locally}). One interesting question is whether $O(1)$-independence might suffice. It \emph{does suffice} for the classical $O(x^2)$ worst-case expected bound \cite{PaghPaRu07}, but it is not known whether it suffices for an amortized expected $\tilde{O}(x)$ bound \cite{bender2022linear}.} For given element $u$, we will refer to $h(u)$ as the \defn{hash} of $u$. Similarly, the \defn{hash} of a tombstone is given by $h(u)$, where $u$ is the element that was deleted.

When discussing the array slots $1, 2, \ldots, n$, one must be a bit careful about \defn{wraparound} (i.e., if we go off the end of the array, we wrap around to the start). Wraparound makes it slightly tricky to talk about relationships $i < j$ for slots $i$ and $j$. Thus, it is helpful to perform a slight abuse of notation, as in \cite{bender2022linear}, and to say that if we reached slot $j$ from $i$ by traveling to the right (resp. left) then $i < j$ (resp. $i > j$).

The hash table supports insertions (i.e., add an element that was not already present), queries (i.e., check whether a given element is present), and deletions (i.e., remove an element that is present): 
\begin{itemize}
\item An insertion of an item $ u$ begins by calculating the smallest index $j \ge h(u)$ such that slot $j$ is either a free slot or contains an element/tombstone with hash at least $j$. This is the position where $u$ will ultimately be placed. Next, we calculate the smallest index $j' \ge j$ such that slot $j'$ is either a free slot or a tombstone. Finally, we slide the elements in positions $j, j + 1, \ldots, j' - 1$ into positions $j + 1, j + 2, \ldots, j'$ (if $j' = j$, then this is a no-op), and we place $ u $ in position $ j $. 
\item A query of an item $u$ examines slots $h(u), h(u) + 1, \ldots$ until it either (1) finds $u$, meaning that $u$ is present; (2) encounters a free slot, meaning that $u$ is not present; or (3) encounters an element/tombstone with hash larger than $h(u)$, meaning once again that $u$ is not present.
\item Finally, a deletion of an item $u$ simply marks the item as deleted, i.e., replaces it with a tombstone. Additionally, for every $ R $ insertions/deletions that occur, a \defn{rebuild} is performed in which the hash table is reconstructed to remove all of the tombstones (the hash function does not change). The time between two rebuilds is referred to as a \defn{rebuild window}, and the parameter $ R $ is referred to as the \defn{rebuild window size}.
\end{itemize}

We say the hash table is filled to $\delta$ \defn{full} (a.k.a., the hash table has \defn{load factor} $\delta$) if it contains $\delta n$ elements (tombstones do not count). We will typically use $1 - 1/x$ to denote the maximum load factor that the hash table ever reaches. 

\paragraph{The path surplus problem.}
The \defn{Path Surplus Problem} is defined as follows. Consider the grid $[0, m] \times [0, m]$ (which we will often refer to simply as $[m] \times [m]$), and suppose that the grid contains $\Pois(m^2)$ blue dots placed uniformly at random, and $\Pois(m^2)$ red dots also placed uniformly at random. In other words, both blue and red dots are placed according to Poisson arrivals with density $1$.

Now, consider the set of \defn{monotonic} paths through the grid, going from $(0, 0)$ to $(m, m)$ (i.e., paths that travel upwards and to the right). As a convention, we will refer to these simply as \defn{paths} (so the fact that they are monotonic, and that they go from $(0, 0)$ to $(m, m)$ will be implicit). Define the \defn{surplus} of a path to be the number of blue dots beneath the path, minus the number of red dots beneath the path. 

The \defn{Path Surplus Problem} is to solve for the expected maximum surplus of any path. That is, we wish to solve for the expected value of the largest surplus of any monotonic path going from $(0, 0)$ to $(m, m)$. An example is given in Figure \ref{fig:surplus}.

\paragraph{The relationship between linear probing and path surplus.}
It is worth taking a moment to briefly describe the relationship between linear probing and the path surplus problem. 

Suppose we wish to analyze the amortized expected insertion cost during a rebuild window spanning time $[t_0, t_1)$. Rather than directly analyzing the cost of each insertion, one can instead analyze the \defn{crossing numbers} $c_1, c_2, \ldots, c_n$, where $c_i$ is the number of times that an insertion with a hash smaller than $ i $ uses either (a) a tombstone left by a key that had hash at least $ i $; or (b) a free slot in a position greater than or equal to $ i $. To a first approximation, $\sum_i c_i$ will be proportional to the total cost of all the insertions in the rebuild window.

The crossing number $ c_i $, in turn, has a surprisingly natural geometric interpretation. Consider a subinterval of the array of the form $[j, i - 1]$ for some $j < i$, and define the \defn{surplus} of the subinterval as follows. Construct a two-dimensional plot of the insertions/deletions with hashes in $[j, i - 1]$, where insertions are represented by blue dots, deletions are represented by red dots, and the position of the dot for an operation on a key $u$ at time $t$ is given by 
$$\left(\frac{h(u) - j}{i - j} , \frac{t - t_0}{t_1 - t_0}\right) \in [0, 1) \times [0, 1).$$
Then the \defn{surplus} of the subinterval $[j, i - 1]$ is simply the maximum surplus of any monotonic path through the grid traveling from $(0, 0)$ to $(1, 1)$. 

With this definition in place, it turns out that the crossing number $c_i $ is \emph{exactly} equal to
\begin{equation}\max_{j < i} \left(\surplus([j, i - 1]) - \free([j, i - 1])\right),
\label{eq:crosssurplus}
\end{equation}
where $\free([j, i - 1])$ denotes the number of free slots in $[j, i - 1]$ at time $t_0$. What is going on intuitively is that, if there is some interval $[j, i - 1]$ with a path $P$ of surplus $s$, and if there are $f$ free slots in the interval at time $t_0$, then the only way for an insertion (i.e., blue dot) beneath $P$ to \emph{avoid} contributing to $c_i$ is if the insertion either (1) makes use of a tombstone created by a deletion (i.e., red dot) beneath $P$, or (2) makes use of one of the $f$ free slots. Thus, the crossing number $c_i$ is guaranteed to be at least $s - f$. So paths with large surplus act as certificates for the crossing number $c_i$, and it turns out that the relationship is tight, hence \eqref{eq:crosssurplus}.

The problem of solving for the surplus of a given interval $[j, i - 1]$ is similar to but not exactly the same as the path surplus problem. Nonetheless, by first solving the path surplus problem, we will be able to obtain a solution that extends to analyzing the surplus of each interval, which will then lead to an amortized analysis of linear probing. 

\paragraph{The techniques known so far.} Finally, we conclude the section with a brief discussion of Bender et al.'s techniques for analyzing path surplus \cite{bender2022linear}. Although our approach will be quite different, the techniques are nonetheless helpful to get a bearing on the problem. 

To simplify the discussion, let us imagine that there are exactly $m^2$ dots, each of which is randomly either blue or red, and that they form a grid-like pattern. Let $\mathcal{P}$ be the set of $\binom{2m}{m} = 2^{\Theta(m)}$ different monotonic paths. The set $\mathcal{P}$ is large enough that we cannot hope to apply a union bound over all paths. We must instead capture the fact that, even though there are many paths, their surpluses are tightly correlated. 

Bender et al.~prove their upper bound as follows. They first consider the $\poly(m)$ axis-aligned \emph{rectangles} in the grid. With a simple Chernoff-bound argument, they show that, with high probability, every rectangle $R$ has at most $O(\operatorname{Perimeter}(R) \cdot \log m)$ more blue dots than red dots in its interior. They then show how to decompose the area underneath any given path $P$ into a collection of rectangles $R_1, R_2, \ldots$ whose perimeters sum to $O(m \log m)$. Summing over the rectangles gives a bound of $O(m \log^2 m)$ on the maximum surplus.

They then prove a lower bound of $\Omega(m \sqrt{\log \log m})$ with a different argument. Consider some threshold $T = \Theta(\sqrt{\log \log m})$. Start with the trivial path that covers all points (i.e., goes from $(0, 0)$ to $(0, m)$ to $(m, m)$). If the top-left \emph{quadrant} of the grid has surplus at least 
$\Omega(m T)$, then we accept this trivial path as our final solution. Otherwise, we remove the quadrant from the path, we break the path in half (one half traveling over the bottom left quadrant, and one half traveling over the top right quadrant) and we recursively apply the same construction to the two halves (using a $m/2$ as the subproblem size in place of $m$). The threshold $T$ is selected so that, within the $O(\log m)$ problems that one encounters on a given recursive path (before hitting a leaf subproblem of size $1$), there is a constant probability that at least one of them will succeed. The expected sum of the sizes of the successful subproblems is therefore $\Theta(m)$, which gives an expected path surplus of $\Theta(m \cdot T) = \Theta(m \sqrt{\log \log m})$. 

The bound of $O(m \log^{0.75} m)$ in the current paper will be achieved through a very different (and somewhat more intricate) set of techniques. Our lower-bound construction, in particular, will reveal a surprising (and beautiful) connection between the structure of the surplus-maximizing path, and a certain type of random walk. We will then be able to prove that this construction is tight through a sequence of potential-function arguments that allow us to directly analyze the ``added value'' that a path could hope to get by deviating from our lower-bound construction.

\section{The Path Surplus Problem}\label{sec:pathsurplus}

In this section, we give tight bounds for the path surplus problem. 
Section \ref{sec:pathsurpluslower} shows that the expected maximum path surplus is $\Omega(m \log^{0.75} m)$; and Section \ref{sec:pathsurplusupper} gives a matching upper bound of $O(m \log^{0.75} m)$.

\begin{figure}
\begin{center}\includegraphics[scale=0.6]{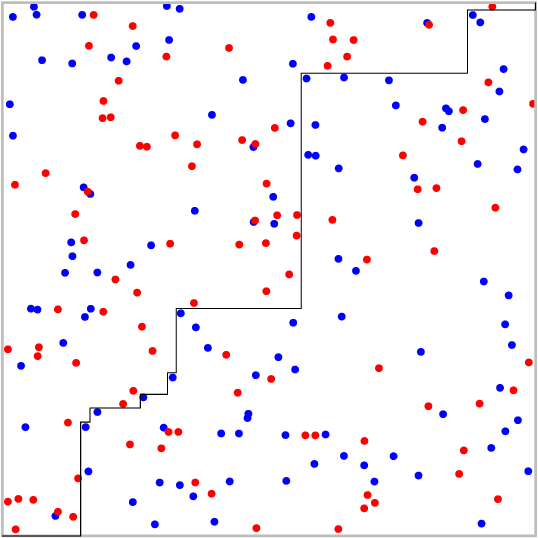}
\end{center}
\caption{An example with $m^2 = 100$. In this case, there are 110 blue dots and 97 red dots. A surplus-maximizing monotonic path is given, and the surplus of the path is 32.}
\label{fig:surplus}
\end{figure}
\subsection{Lower Bound on Path Surplus}\label{sec:pathsurpluslower}

To describe our lower-bound construction, it will be helpful to rotate the problem by 45 degrees, as in Figure \ref{fig:rotated}. Now, the path travels from $(0, 0)$ to $(m', 0)$, where \(m' = \sqrt{2} m\). The original restriction that the path be monotone now translates to the path having slope in \([-1, 1]\) at all times. Call such a path \defn{slope-legal}. This is the perspective that we will take for the entirety of this subsection.

As notation, for any given line segment $L$, let $p_L = (x_L, y_L)$ denote its midpoint, $w_L$ denote its width, and $q_L = w_L / \sqrt{\log m}$. Finally, let $R_L$ denote the rectangular region  
$$\{(x, y) \mid (x, y + s) \in L \text{ for some } |s| \le q_L / 2, \text{ and } |x - x_L| \le w_L / 16\}$$
consisting of points $(x, y)$ that are within vertical distance $q_L/2$ of $L$ and horizontal distance $w_L/16$ of $p_L$. Again, throughout the subsection, when we discuss vertical and horizontal distances, we are referring to distances in the rotated version of the problem (i.e., what were originally diagonal distances).

Let \(L_0\) be the straight line path from $(0, 0)$ to $(m', 0)$. Our path \(P\) is constructed by calling the recursive function \(\texttt{Path}(L_0, 1)\) (see Algorithm \ref{alg:path}). The function takes as input a line-segment $L$, and moves the midpoint $p_L = (x_L, y_L)$ to $p'_L = (x_L, y_L \pm q_L)$, where the $\pm$ is determined by whether $\surplus(R_L)$ is positive or negative. If $\surplus(R_L)$ is positive, then $p'_L = (x_L, y_L + q_L)$, so that $R_L$ is contained beneath the path, and if $\surplus(R_L)$ is negative, then $p'_L = (x_L, y_L - q_L)$, so that $R_L$ is contained above the path. (If $\surplus(R_L) = 0$, a random decision is made.) The act of moving the midpoint breaks $L$ into two line segments $A_L$ and $B_L$, which we then recurse on. The recursion stops when either we reach depth $(\log m) / 16$ or we get to a line segment $L$ whose slope is very close to $1$ or $-1$ (within less than $1 / \sqrt{\log m}$). An example of what the variables $L, R_L, p_L, A_L, B_L$ might look like for the base case $L_0 = L$ is given in Figure \ref{fig:basecase}.

\begin{algorithm}
\caption{Given a line segment \(L\), and a recursion depth \(d\), we split the line segment into two segments, shift the middle point vertically by some amount, and recurse on the two segments.}
\begin{algorithmic}[1]
\Procedure{Path}{$L, d$}
\If{$d \le (\log m) / 16$ and $-1 + 1/\sqrt{\log m} < \slope(L) < 1 - 1 / \sqrt{\log m}$}
\State Let \(p_L = (x_L, y_L)\) be the midpoint of \(L\).
\State Let \(w_L\) be the width of \(L\), that is, the horizontal distance that \(L\) travels.
\State Let \(q_L = w_L / \sqrt{\log m}\).
\State Let \(R_L = \{(x, y) \mid (x, y + s) \in L \text{ for some } |s| \le q_L / 2, \text{ and } |x - x_L| \le w_L / 16\}\).

\If{$\surplus(R_L) > 0$}
\State Let \(p'_L = (x_L, y_L + q_L)\). 
\EndIf
\If{$\surplus(R_L) < 0$}
\State Let \(p'_L = (x_L, y_L - q_L)\). 
\EndIf
\If{$\surplus(R_L) = 0$}
\State Set \(p'_L = (x_L, y_L \pm q_L)\) for a random choice of \(\pm \in \{+, -\}\). 
\EndIf
\State Let \(A_L\) go from the left end-point of \(L\) to \(p'_L\).
\State Let \(B_L\) go from the right end-point of \(L\) to \(p'_L\).
\State Replace \(L\) with \(A_L, B_L\).
\State Call \(\texttt{Path}(A_L, d + 1)\).
\State Call \(\texttt{Path}(B_L, d + 1)\).
\EndIf
\EndProcedure
\end{algorithmic}
\label{alg:path}
\end{algorithm}

We can capture the behavior of Algorithm \ref{alg:path} with a binary tree \(T\) whose nodes are line segments. The root of the tree is the straight line path \(L_0\). For any node \(L\) in the tree, if the recursion terminates on \(L\) then \(L\) is a leaf, and otherwise \(L\) has children \(A_L\) and \(B_L\) as constructed by the algorithm. 

In order to describe the analysis, it will be helpful to introduce two additional definitions. For a given node $L \in T$, call a point \((x, y) \in \mathbb{R}^2\) \defn{\(L\)-critical} if $L$ is not a leaf, and if $(x, y) \in R_L$; and call a point \((x, y)\) \defn{\(L\)-sensitive} if $L$ is not a leaf, and if it is possible for the point to be in the triangle with edges \(L, A_L, B_L\) (for some outcome of $\text{sign}(\surplus(R_L))$). All \(L\)-critical points are also \(L\)-sensitive, but not vice versa.

A key feature of the algorithm is that $L$-critical points at one node $L$ cannot be $L'$-sensitive for any descendent $L'$ (Lemma \ref{lem:sensitivity}). This has two implications: (1) that the rectangles $R_L$ and $R_{L'}$ are disjoint, meaning that the random bits determining the behaviors of the two subproblems are independent; and (2) that the question of whether $R_L$ is above/below the final path is completely determined by the decision made at the subproblem corresponding to node $L$, and not affected by subsequent decisions made in descendent subproblems.
\begin{lemma}
If a point \((x, y) \in \mathbb{R}^2\) is \(L\)-critical for some node $L \in T$, then it cannot be \(L'\)-sensitive (or \(L'\)-critical) for any node $L' \in T$ that is a descendent of $L$.
\label{lem:sensitivity}
\end{lemma}
\begin{proof}
Suppose for contradiction that \((x, y)\) is \(L'\)-sensitive. Let \(L_1, L_2, \ldots, L_t\) be the tree path from \(L_1 = L\) to \(L_t = L'\) for some \(t\). 
Because $(x, y)$ is $L_1$-critical, we have that
\begin{equation}|x_{L_1} - x| \le w_{L_1} / 16 = w_{L_2} / 8.
\label{eq:firstgeomin}
\end{equation}
It follows that
\begin{equation}|L_1(x) - L_2(x)| \ge \frac{7}{8} q_L.\label{eq:L2}\end{equation} Recalling that $(x_{L_2}, y_{L_2})$ denotes the midpoint of $L_2$, we have as another consequence of \eqref{eq:firstgeomin} that
$|x_{L_2} - x| \ge \frac{3}{4} w_{L_3}$,
which implies that
\begin{equation}|L_2(x) - L_3(x)| \le \frac{1}{4} q_{L_2}.
\label{eq:L3}
\end{equation}
Combined, \eqref{eq:L2} and \eqref{eq:L3} imply that
\begin{equation}|L_3(x) - L_1(x)| \ge \frac{7}{8}q_{L_1} - \frac{1}{4} q_{L_2} = \frac{7}{8}q_{L_1} - \frac{1}{8} q_{L_1} = \frac{3}{4} q_{L_1}.
\label{eq:L13}
\end{equation}
On the other hand, for all \(i \ge 3\),
$$|L_i(x) - L_{i + 1}(x)| \le q_{L_i} \le \frac{1}{2^i q_L}.$$
It follows that
$$|L_t(x) - L_3(x)| < \sum_{i \ge 3} \frac{1}{2^i q_L} \le \frac{1}{4} q_L.$$
Combining this with \eqref{eq:L13}, we get that
$|L_1(x) - L_t(x)| > \frac{1}{2} q_{L_1}.$
Finally, this implies that $|L_1(x) - L_t(x)| > q_{L_t}$, which contradicts the fact that $x$ is $L_t$-sensitive.
\end{proof}

Having established that each point $(x, y)$ is critical to at most one node in the tree $T$, we can now reason about the probability distribution for $\surplus(R_L)$ at a given node $L$ in the tree. The quantity that we care about, in particular, is $\surplus(R_L) \cdot \mathbb{I}_{R_L \ge 0}$, since this is the surplus that $R_L$ will contribute to our final path (recall that if $\surplus(R_L) < 0$ then it will not be contained beneath the final path). 
\begin{lemma}
    Consider a non-leaf node \(L\) at depth \(d\) in \(T\). Fix outcomes for \(\surplus(R_{L'})\) for every node \(L'\) at depths less than \(d\), and let \(\mathcal{O}\) denote these outcomes. Then
    $$\E[\surplus(R_L) \cdot \mathbb{I}_{R_L \ge 0} \mid \mathcal{O}] \ge \Omega\left(\frac{m}{2^d \log^{0.25} m}\right).$$
    \label{lem:rectanglesurplus}
\end{lemma}
\begin{proof}
    By Lemma \ref{lem:sensitivity}, the region \(R_L\) is disjoint from the regions \(R_{L'}\) for other nodes \(L'\). The distribution of \(\surplus(R_L)\) is therefore unaffected by \(\mathcal{O}\).

    We can express \(\surplus(R_L)\) as \(A - B\), where \(A\) and \(B\) are the numbers of blue and red points in \(R_L\), respectively. The random variables \(A\) and \(B\) are each independently distributed according to \(\operatorname{Pois}(\alpha)\), where \(\alpha\) is the area of the geometric region \(R_L\). Observe that
    $$\alpha = \Theta(q_L w_L) = \Theta\left(\frac{m}{2^d} \cdot \frac{m}{2^d \sqrt{\log m}}\right) = \Theta\left(\frac{m^2}{4^d \sqrt{\log m}}\right).$$
    With probability \(\Omega(1)\), we have that \(A \ge \alpha + \sqrt{\alpha}\) and that \(B \le \alpha\), in which case \(A - B \ge \sqrt{\alpha}\). Thus
    $$\E[\surplus(R_L) \cdot \mathbb{I}_{R_L \ge 0} \mid \mathcal{O}] \ge \Omega(\sqrt{\alpha}) = \Omega\left(\frac{m}{2^d \log^{0.25} m}\right).$$
\end{proof}

Lemma \ref{lem:rectanglesurplus} captures the contribution of $R_L$, for a given node $L$, to the final path surplus. However, in order to guarantee that the final path surplus is large, we must also ensure that there are many nodes $L$ in each level of the tree. Recall that even nodes $L$ with relatively small depths can be leaves if $|\slope(L)|$ is very near $1$. We must show that this is relatively rare---that is, that we expect each level $d \le (\log m)/16$ in the tree to be nearly saturated. 
\begin{lemma}
Let \(\tau_d\) be the number of nodes at depth \(i\) in tree \(T\). For any \(d \le (\log m) / 16\),
$\E[|\tau_d|] \ge \Omega(2^d).$
\label{lem:levelsize}
\end{lemma}
\begin{proof}
Let \(V = \langle v_1, v_2, \ldots, v_{d - 1}\rangle \in \{1, -1\}^{d - 1}\) be a vector of \(\pm 1\)s. 
Consider the tree path \(Q_V\) that starts at the root, and travels for (up to) \(d - 1\) steps: if, after the \((i - 1)\)-th step the path is not yet at a leaf, then the \(i\)-th step goes to a child of the current node, chosen by \(v_i\) (left child if \(v_i = 1\), and right child if \(v_i = -1\)). Say that \(Q_V\) \defn{terminates early} if it reaches a leaf at depth \( < d\).

We will show that, with probability \(\Omega(1)\), \(Q_V\) does not terminate early. It ten follows that
$$\E[|\tau_d|] \ge \sum_{V \in \{1, -1\}^{d - 1}} \Pr[Q_V \text{ does not terminate early}] \ge \Omega(2^d).$$

To analyze the probability of \(Q_V\) terminating early, consider the slopes \(s_1, s_2, \ldots, s_j\), \(j \le d\), of the line segments \(L_1, L_2, \ldots, L_j\) (i.e., nodes) that \(Q_V\) encounters as it traverses down the tree. By construction, \(s_1 = 0\), and each \(s_{i + 1}\) is randomly one of 
\begin{equation}s_i \pm \frac{1}{\sqrt{\log m}}.\label{eq:pm}\end{equation}
Recall that the choice of \(\pm\) in \eqref{eq:pm} is determined by the number of blue/red points that are \(L_i\)-critical. By Lemma \ref{lem:sensitivity}, the \(L_i\)-critical region is disjoint from the \(L_k\)-critical region for all \(i \neq k\). Thus the \(s_i\)s are performing an unbiased random walk where each step is independent of the previous ones.

If \(Q_V\) terminates early, then at least one of \(s_1, s_2, \ldots, s_j\) must be within \(1/\sqrt{\log m}\) of \(1\) or \(-1\). Or, to use a simpler condition, at least one of \(s_1, s_2, \ldots, s_j\) must leave the interval \([-0.5, 0.5]\). 

Thus, we can bound the probability of \(Q_v\) terminating early by the probability that the random walk \(s_1, s_2, \ldots\) leaves the interval \([-0.5, 0.5]\) in its first \(\kappa = (\log m)/16\) steps. The random walk starts at \(0\) and moves randomly by $$\pm \frac{1}{\sqrt{\log m}} = \frac{1}{4\sqrt{\kappa}}$$ on each step. By Kolmogorov's inequality (i.e., Chebyshev's inequality for martingales), the probability that the random walk leaves the interval \([-0.5, 0.5]\) is at most \(1/4\). It follows that the probability of \(Q_V\) terminating early is at most \(1/4\).
\end{proof}

\begin{corollary}For \(d \le (\log m) / 16 - 1\), the expected number of internal depth-$d$ noes in $T$ is $\Omega(2^d).$
\label{cor:levelsize}
\end{corollary}

We can now analyze the surplus of the path produced by Algorithm \ref{alg:path}.
\begin{theorem}
   Let \(L_0\) be the straight-line path from \((0, 0)\) to \((\sqrt{2}m, 0)\). Let \(P\) be the path produced by \(\texttt{Path}(L_0, 1)\) (defined in Algorithm \ref{alg:path}). Then, \(P\) is a slope-legal path with expected surplus \(\Omega(m \log^{0.75} m)\). 
\end{theorem}
\begin{proof}
   The fact that \(P\) is slope-legal follows from the fact that \(\texttt{Path}(L, d)\) only recurses if \(-1 + 1 / \sqrt{\log m} < \slope(L) < 1 - 1 / \sqrt{\log m}\). This means that \(\texttt{Path}(L, d)\) is only ever called on paths \(L\) satisfying \(-1 < \slope(L) < 1\). 

   Let \(\mathcal{R}\) denote the set of rectangles \(R_L\) for internal nodes \(L \in T\). Note that, by Lemma \ref{lem:sensitivity}, any \(R_L\) satisfying \(\surplus(R_L) > 0\) will be contained beneath \(P\), and any \(R_L\) satisfying \(\surplus(R_L) < 0\) will be contained above \(P\). Define the \defn{partial surplus} \(\surplus'(P)\) to be 
   $$\surplus'(P) = \sum_{R \in \mathcal{R}} \surplus(R) \cdot \mathbb{I}_{\surplus(R) > 0}.$$
   The surplus of \(P\) can be expressed as
   $\surplus(P) = \surplus'(P) + \surplus(G),$
   where \(G\) is the region underneath \(P\) consisting of points not in \(\cup_{R \in \mathcal{R}}R\). Since \(\E[\surplus(G)] = 0\), it suffices to prove that \(\E[\surplus'(P)] \ge \Omega(m \log^{0.75} m)\).

   Let \(\tau'_d\) be the number of internal nodes at depth \(d\) in \(T\). By Lemma \ref{lem:rectanglesurplus},
   $$\E[\surplus'(P)] \ge \Omega\left(\E\left[\sum_d \tau'_d \cdot \frac{m}{2^d \log^{0.25} m} \right]\right).$$
    By Corollary \ref{cor:levelsize}, this is at least
   $$\Omega\left(\sum_{d < (\log m)/16 - 1} 2^d \cdot \frac{m}{2^d \log^{0.25} m}\right) = \Omega(m \log^{0.75} m).$$
\end{proof}

\begin{figure}
\begin{center}\includegraphics[scale=0.6]{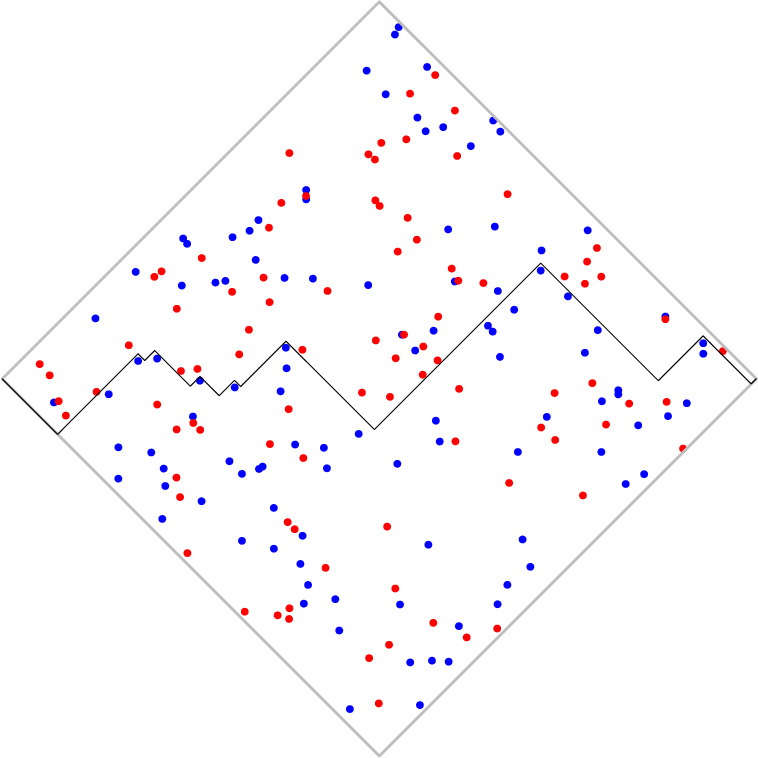}
\end{center}
  \caption{The same path as in Figure \ref{fig:surplus}, but in the rotated version of the problem considered in Section \ref{sec:pathsurpluslower}. The constraint that the path is monotonic now becomes a constraint on slope: the slope of the path must always stay in $[-1, 1]$.}  
  \label{fig:rotated}
\end{figure}

\begin{figure}
\begin{center}\includegraphics[scale=0.6]{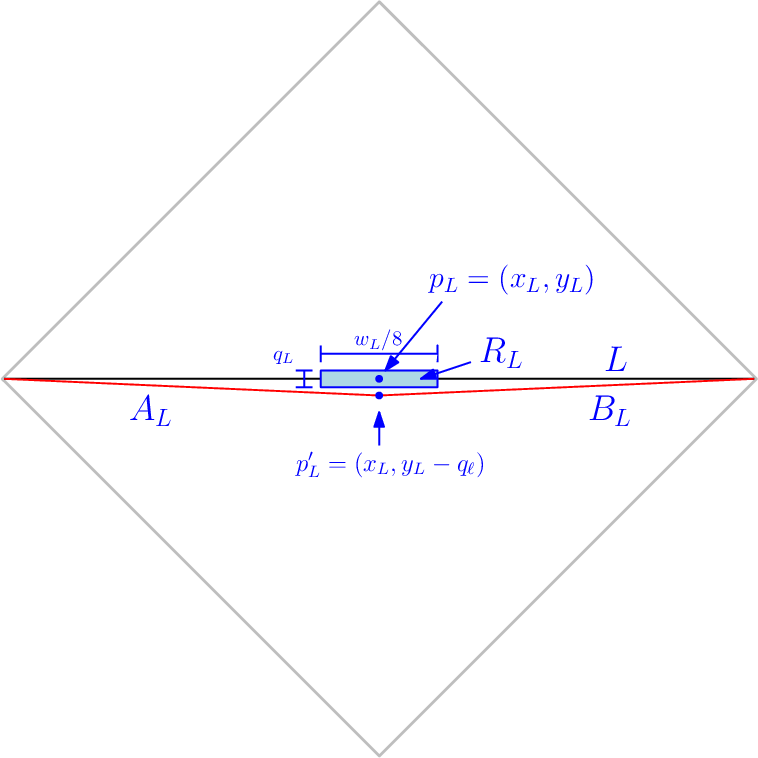}
\end{center}
  \caption{An example of what the base-case subproblem (i.e., $L = L_0$) would look like if $\surplus(R_L) < 0$ (so $p'_L = (x_L, y_L - q_L)$). The algorithm would then recurse on $A_L$ and $B_L$. Note that $A_L$ and $B_L$ have slopes that are very close (within $O(1/ \sqrt{\log m})$) to that of $L$. This will be important for making sure that the recursion is able to (most likely) get to depth $\Theta(\log m)$ before terminating (i.e., before getting to a line with slope close to $1$).}
    \label{fig:basecase}
\end{figure}

\subsection{Upper Bound on Path Surplus}\label{sec:pathsurplusupper}


Next, we turn to the problem of proving an upper bound. We will show that, both in expectation and with good probability, the maximum surplus of any path is $O(m \log^{0.75} m)$.

\paragraph{Three core facts.} We begin by explicitly stating the three `core facts' that we will use in the analysis. These are the only facts that we will use about the placement of blue/red dots. By making the core facts explicit, we will be able to subsequently (in Section \ref{sec:insertionsurpluslower}) extend our analysis to other settings where the blue/red dots arrive via a more complicated combinatorial process.

\begin{lemma}[Core Fact 1: Vertical Independence]
Partition the $[m] \times [m]$ grid into disjoint vertical strips $V_1, V_2, \ldots$, and define $B_i$ and $R_i$ to be the sets of blue and red dots in strip $V_i$, respectively. Then the pairs $(B_1, R_1), (B_2, R_2), \ldots$ are mutually independent random variables.
\label{lem:fact1}
\end{lemma}

\begin{lemma}[Core Fact 2: Region Surplus]
Consider any geometric region $G$, and let $A$ be the area of $G$. Let $B$ (resp.~$R$) denote the number of blue dots (resp.~red dots) that appear in $G$. Then,
\begin{equation}\Pr[|B - R| \ge \kappa] \le e^{-\Omega(\kappa^2 / A)} + e^{-\Omega(\kappa)}.\end{equation}
\label{lem:fact2}
\end{lemma}
\begin{proof}
As both $B$ and $R$ are Poisson random variables satisfying $\E[B] = \E[R] = A$, the lemma follows by applying Chernoff bounds to each of $B$ and $R$ to bound the probability that either deviates from its mean by at least $\kappa / 2$.
\end{proof}

\begin{lemma}[Core Fact 3: Restricted-Path Upper Bound]
Consider any geometric region $G$, and let $A$ be the area of $G$. Let $\mathcal{P}_G$ be the set of monotonic paths that stay within $G$ at all times. Define
$$S = \max_{P \in \mathcal{P}_G} \surplus(P) - \min_{P \in \mathcal{P}_G} \surplus(P).$$
For $\alpha \ge 1$,
\begin{equation*}\Pr[S \ge \alpha A] \le e^{-\Omega(\alpha) \cdot A}.
\end{equation*}
\label{lem:fact3}
\end{lemma}
\begin{proof}
    We can upper-bound $S$ by the number of dots in $G$, which is a Poisson random variable with mean $2A$. The result follows from a Chernoff bound.
\end{proof}

\paragraph{Notation and conventions for the section. }In addition to the core facts, let us also establish some notation and conventions to be used throughout the section. Define $d_\ell = m / 2^\ell$ and $q_\ell = m / (2^\ell \sqrt{\log m})$. Let $\overline{\ell} = \log \left(m / \sqrt{\log m}\right)$ be the index at which $q_{\overline{\ell}} = 1$.

For $\ell \in [\overline{\ell}]$, define the \defn{level-$\ell$ diagonals} to be the diagonals $D^{(\ell)}_0, D^{(\ell)}_1, \ldots, D^{(\ell)}_{2^\ell}$ such that $D^{(\ell)}_j$ contains the points $(x, y)$ satisfying $x + y = 2j d_\ell$. (Note that we are no longer using the rotated coordinate system defined in the previous subsection.) Define the \defn{$\ell$-coordinate vector} of a path $P$ to be the vector $\coord_\ell(P) = \langle r_0, r_1, \ldots, r_{2^\ell} \rangle$ such that $P$ intersects $D^{(\ell)}_j$ at point $\langle j d_\ell, j d_\ell \rangle + r_j \langle 1, -1\rangle$. (In particular, $r_0$ and $r_{2^\ell}$ are necessarily $0$ since the path begins at $\langle 0, 0\rangle$ and ends at $\langle 1, 1 \rangle$.) For any $a, b \in \mathbb{R}$, let $\round(a, b)$ be $a$ rounded to the next multiple of $b$. Define the \defn{$\ell$-frame} 
$\frame_\ell(P)$ to be the vector $\langle \round(r_0, q_\ell), \round(r_1, q_\ell), \ldots, \round(r_{2^\ell}, q_\ell) \rangle$, where $r_i$ is the $i$-th coordinate of $\coord_\ell(P)$. Define the \defn{refined $\ell$-frame} to be the vector 
$$\frame'_\ell(P) = \langle \round(r_0, q_{\ell + 1}), \round(r_1, q_{\ell + 1}), \ldots, \round(r_{2^\ell}, q_{\ell + 1}) \rangle,$$ where again $r_i$ is the $i$-th coordinate of $\coord_\ell(P)$. (By definition, the $i$-th coordinate of 
$\frame'_\ell(P)$ is equal to the $2i$-th coordinate of $\frame_{\ell + 1}(P)$.) 

Given $v = \langle r_0, r_1, \ldots, r_{2^\ell}\rangle$ that is either an $\ell$-frame or a refined $\ell$-frame, define the \defn{implied path} $\operatorname{Path}(v)$ of the vector to be the path $P$ satisfying $\coord_\ell(P) = v$ and such that $P$ follows straight lines between the diagonals $D^{(\ell)}_0, D^{(\ell)}_1, \ldots, D^{(\ell)}_{2^\ell}$. Define the \defn{surplus} of an $\ell$-frame (and, similarly, of a refined $\ell$-frame) to be the surplus of the frame's implied path.

Finally, for $\ell < \overline{\ell}$, define the \defn{$\ell$-transition vector} $\trans_\ell(P) = \langle t_0, t_1, \ldots, t_{2^{\ell} - 1} \rangle$ to be given by $$t_i = \Big\lceil \frac{b_{2i + 1} - \round(a_i/2 + a_{i + 1}/2, q_{\ell + 1})}{q_{\ell + 1}} \Big\rceil,$$ where $a_j$ is the $j$-th coordinate of $\frame'_\ell(P)$ and $b_j$ is the $j$-th coordinate of $\frame_{\ell + 1}(P)$. The way to think about the quantity $\round(a_i/2 + a_{i + 1}/2, q_{\ell + 1})$ is that it is what the $(2i + 1)$-th coordinate of $\frame_{\ell + 1}(P)$ would be if $P$ were simply the implied path $\operatorname{Path}(F'_\ell)$; and, therefore, the way to think about $t_i$ is that it measures (in multiples of $q_{\ell + 1}$) the difference between the true $(2i + 1)$-th coordinate of $\frame_{\ell + 1}(P)$ versus the $(2i + 1)$-th coordinate of $\frame_{\ell + 1}(\operatorname{Path}(F'_\ell))$. Given $\frame'_\ell(P)$ and $\trans_\ell(P)$, one can recover $\frame_{\ell + 1}(P)$, and similarly, given $\frame'_\ell(P)$ and $\frame_{\ell + 1}(P)$, one can recover $\trans_\ell(P)$.

We can summarize the relationships between $\ell$-frames, refined $\ell$-frames, and transition vectors as follows. An $\ell$-frame $F$ looks at $2^\ell + 1$ evenly-spaced diagonal lines and keeps track of where the path hits each diagonal line at a granularity of $\sqrt{2} q_\ell = \sqrt{2} m / (2^\ell \sqrt{\log m})$.\todo{Maybe explicitly explain what the $\sqrt{2}$ is doing} A refined $\ell$-frame looks at the same $2^\ell$ evenly-spaced diagonal lines, but keeps track of where the path hits them at a slightly finer granularity of $\sqrt{2} q_{\ell + 1} = \sqrt{2} q_\ell / 2$. Finally, the $\ell$-transition vector tells us what we need to know to get from the refined $\ell$-frame $F_\ell$ to the $(\ell + 1)$-frame $F_{\ell + 1}$ for some path.

We say that a given $\ell$-frame $F$ and refined $\ell$-frame $F'$ are \defn{compatible} if there is a path $P$ whose $\ell$-frame is $F$ and whose refined $\ell$-frame is $F'$. Similarly, we can talk about refined $\ell$-frames being compatible with $(\ell + 1)$-frames, refined $\ell$-frames being compatible with $\ell$-transition vectors, etc.

 Given an $\ell$-frame $F_\ell$ and a refined $\ell$-frame $F'_\ell$, define the \defn{refinement gain} $\gain(F_\ell, F'_\ell) = \surplus(F'_\ell) - \surplus(F_\ell)$. Given a refined $\ell$-frame $F'_\ell$, an $\ell$-transition vector $T_\ell$, and an $(\ell + 1)$-frame $F_{\ell + 1}$ that are all compatible with each other, define the \defn{transitional gain} $\gain(F_\ell, T_\ell) = \surplus(F_{\ell + 1}) - \surplus(F'_\ell)$. As a shorthand, we will sometimes use $\gain(F_\ell, F'_\ell, T_\ell)$ to denote $\gain(F_\ell, F'_\ell) + \gain(F'_\ell, T_\ell)$.
 

\paragraph{Relationship to the construction in Section \ref{sec:pathsurpluslower}.}
It is also worth taking a moment to understand what the $\ell$-transition vectors would look like for a path constructed as in the lower-bound construction from Section \ref{sec:pathsurpluslower}. Recall that the construction starts with a straight-line diagonal path, and then recursively: (1) splits the line segment into two segments $A_L$ and $B_L$; (2) slides the mid-point of the segments by some amount; and (3) recurses on the two segments. In the $i$-th level of recursion, the subproblems of the recursion correspond to line segments $L$ between consecutive level-$\ell$ diagonals. The mid-point between $A_L$ and $B_L$ lies on a level-$(\ell + 1)$-diagonal. The construction slides the midpoint along the diagonal by $\pm \Theta(m / (2^\ell \sqrt{\log m})) = \pm \Theta(q_{\ell + 1})$. This corresponds to using a transition vector with entries $\pm O(1)$. 

In other words, in the context of the lower-bound construction from Section \ref{sec:pathsurpluslower}, all of the transition vectors have entries of the form $\pm O(1)$. Of course, in general, the entries of a given transition vector could be as large as $O(\sqrt{\log m})$. One of the main obstacles in proving the upper bound will be to show that, to a first approximation, if a path wishes to have a large surplus, then there is no asymptotic advantage to ever having super-constant transition-vector entries.

\paragraph{The analysis.}
A natural approach to bounding the maximum surplus of any path would be to first bound the maximum level-$\ell$ transitional gain
\begin{equation}\max_P \gain(\frame_\ell(P), \trans_\ell(P)),
\label{eq:transgain}
\end{equation}
where the maximum is taken over all paths $P$. If we were to only consider paths using transition vectors with $\pm O(1)$ entries, then it turns out we would be able to bound \eqref{eq:transgain} by $O(m / \log^{0.25} m)$. Summing over the $O(\log m)$ levels would bound the sum of the transitional gains across all levels by $O(m \log^{0.75} m)$. The problem with this approach is that, if a transition vector $\trans_\ell(P)$ has super-constant entries, then it can actually cause $\gain(\frame_\ell(P), \trans_\ell(P))$ to be significantly larger than $O(m / \log^{0.25} m)$. The larger the entries of $\trans_\ell(P)$, the larger the amount of area that is under $\operatorname{Path}(\frame_{\ell + 1}(P))$ but not under $\operatorname{Path}(\frame'_\ell(P))$, and thus the larger of an opportunity there is for $\surplus(\frame_{\ell + 1}(P))$ and $\surplus(\frame'_\ell(P))$ to differ. In general, it turns out that \eqref{eq:transgain} could be as large as $O(m)$.

This suggests that we should not be directly bounding $\gain(\frame_\ell(P), \trans_\ell(P))$. Instead, we will examine the quantity
\begin{equation}\gain(\frame_\ell(P), \trans_\ell(P)) - \Phi(\trans_\ell(P)),
\label{eq:adjustedgain}
\end{equation}
where $\Phi$ is a potential function determined by $\trans_\ell(P)$. Define the \defn{transition potential} $\Phi(T)$ of an $\ell$-transition vector $T = \langle t_0, t_1, \ldots, t_{2^{\ell} - 1}\rangle$ to be
$$\Phi(T) = \sum_{i = 0}^{2^\ell - 1} (|t_i| + 1)^2 d_\ell / \log^{0.25} m.$$
The potential function $\Phi(T)$ is larger for transition vectors $T$ with super-constant entries. Thus, intuitively, there should be some hope of obtaining a nice bound for \eqref{eq:adjustedgain}. On the other hand, we can show that for any given path $P$, even though there may be some values of $\ell$ for which $\Phi(\trans_\ell(P))$ is large, the sum of the potentials $\sum_\ell \Phi(\trans_\ell(P))$ is guaranteed to be $O(m \log^{0.75} m)$.
\begin{lemma}
Consider a path $P$, and for each $\ell \in [\overline{\ell} - 1]$, let $T_\ell = \trans_\ell(P)$. Then,
$$\sum_{\ell = 0}^{\overline{\ell} - 1} \Phi(T_\ell)  = O(m \log^{0.75} m).$$
\label{lem:potential1}
\end{lemma}
\begin{proof}
Let $m'$ be $m$ rounded up to the next power of two. Call a path \defn{rigid} if it travels along the edges of an $m' \times m'$ grid. Define $\mathcal{P}$ to be the set of all $\binom{2m'}{m'} \le 4^{m'}$ rigid paths. For $\ell \in [\overline{\ell}]$, define $\mathcal{P}_\ell$ to be the set of rigid paths $Q \in \mathcal{P}$ satisfying $\frame_\ell(Q) = \frame_\ell(P)$. The key to proving the lemma is to show that, for $\ell < \overline{\ell}$,
\begin{equation}
\log |\mathcal{P}_{\ell + 1}| \le \log |\mathcal{P}_{\ell}| - \Omega(\Phi(T_\ell) / \log^{0.75} m) + O(m / \log m).
\label{eq:logrigid}
\end{equation}
It follows that
\begin{align*}
\sum_{\ell = 0}^{\overline{\ell} - 1} \Phi(T_\ell) & \le O\left(\log^{0.75} m \cdot  \sum_{i = 0}^{\overline{\ell} - 1} \left(\log |\mathcal{P}_{\ell}| - \log |\mathcal{P}_{\ell + 1}|\right) + \sum_{i = 1}^{\overline{\ell} - 1} m / \log^{0.25} m\right) \\
                                           & \le O\left(\log^{0.75} m \cdot \log |\mathcal{P}_0| + m \log^{0.75} m\right) \\
                                           & \le O\left(m \log^{0.75} m\right),
\end{align*}
as desired.

It remains to prove \eqref{eq:logrigid}. Fix $\ell < \overline{\ell}$. Let $\mathcal{C}$ denote the set of $\ell$-coordinate vectors $C$ such that $C = \coord_\ell(P)$ for some rigid path $Q \in \mathcal{P}$ satisfying $\frame_\ell(Q) = \frame_\ell(P)$. For each $C \in \mathcal{C}$, let $\mathcal{P}(C)$ denote the set of rigid paths $Q \in \mathcal{P}$ satisfying $\coord_\ell(Q) = C$. The sets $\{\mathcal{P}(C) \mid C \in \mathcal{C}\}$ form a disjoint partition of $\mathcal{P}_\ell$.

For each rigid path $Q \in \mathcal{P}(C)$, let $X_Q$ be the event that $\trans_\ell(Q) = \trans_\ell(P)$, or equivalently, that $Q \in \mathcal{P}_{\ell + 1}$. We will prove that, for each $C \in \mathcal{C}$, and for a random $Q \in \mathcal{P}(C)$,
\begin{equation}
\Pr[X_Q] \le 2^{-\Omega(\Phi(T_\ell) / \log^{0.75} m) + O(m / \log m)}.
\label{eq:Xq}
\end{equation}
Since $\{\mathcal{P}(C) \mid C \in \mathcal{C}\}$ forms a partition of $\mathcal{P}_\ell$, the inequality \eqref{eq:Xq} also holds for a random $Q \in \mathcal{P}_\ell$ (rather than a random $Q \in \mathcal{P}(C)$). Since $\Pr[X_Q] = \Pr[Q \in \mathcal{P}_{\ell + 1}]$, this implies that $|\mathcal{P}_{\ell + 1}| \le |\mathcal{P}_\ell| \cdot  2^{-\Omega(\Phi(T_\ell) / \log^{0.75} m) + O(m / \log m)}$, which gives \eqref{eq:logrigid}.

Thus the lemma reduces to proving \eqref{eq:Xq}. Let $\langle t_0, t_1, \ldots, t_{2^{\ell} - 1}\rangle = \trans_\ell(P)$, and let $\langle t_0', t_1', \ldots, t'_{2^{\ell} - 1} \rangle = \trans_\ell(Q)$. Whereas the $t_i$s are fixed, the $t'_i$s are random variables. Because $Q$ is a random rigid path with $\ell$-coordinate vector $C$, we can think of $Q$ as consisting of independent rigid sub-paths $Q_1, Q_2, \ldots, Q_{2^\ell}$, where $Q_i$ travels from diagonal $D^{\ell}_{i - 1}$ to diagonal $D^{\ell}_i$ (and where the points on the diagonals are determined by $C$).

Now focus on some $Q_i$ (which, in turn, determines $t'_i$), and consider the probability that $t'_i = t_i$. Let $Y$ be the number of vertical steps in the first half of $Q_i$ and $Z$ be the number of horizontal steps. In order for $|t'_i|$ to be at least $k$ for some $k$, we need that
$$|(Y - Z) - \E[Y - Z]| \ge (k - 1) q_\ell.$$
As $Y - Z$ is governed by a Chernoff bound for negatively associated random variables (see, e.g., \cite{wajc2017negative}), we have that
$$\Pr[|t'_i| \ge k] \le 2^{-\Omega( ((k - 1) q_\ell)^2 / d_\ell)} \le 2^{-\Omega((k - 1)^2 d_\ell / \log m)}.$$
Because the $t'_i$s are independent, it follows that
\begin{align*}
\Pr[X_Q] & = \Pr[\trans_\ell(Q) = \trans_\ell(P)] \\    
         & \le \prod_{i = 1}^{2^\ell} \Pr[|t'_i| \ge |t_i|] \\
         & \le \prod_{i = 1}^{2^\ell}  2^{-\Omega((t_i - 1)^2 d_\ell / \log m)} \\
         & \le \prod_{i = 1}^{2^\ell}  2^{-\Omega((t_i + 1)^2 d_\ell / \log m) + O(d_\ell / \log m)} \\
         & \le 2^{-\Omega(\Phi(T)/\log^{0.75} m) + O(m / \log m)}. \\
\end{align*}
This completes the proof of \eqref{eq:Xq}, and therefore the proof of the lemma.    
\end{proof}

With the next two lemmas, we will prove a bound on $\max_T (\gain(F, T) - \Phi(T))$, where $F$ is a fixed refined $\ell$-frame and $T$ ranges over all $\ell$-transition vectors compatible with $F$. Notice that we are not yet taking a maximum over all possible paths, just over all possible transition vectors $T$ for a \emph{given} $\ell$-frame $F$. (Extending this to get a maximum over all paths will actually require quite a few more ideas.)

We begin by isolating the effect of the $j$-th coordinate of the $\ell$-transition vector.
\begin{lemma}
Let $j \in [0, 2^\ell - 1]$. Consider a refined $\ell$-frame $F$, and let $\mathcal{T}$ be the set of $\ell$-transition vectors that are compatible with $F$ and that are $0$ in coordinates $[0, 2^{\ell} - 1] \setminus \{j\}$. Then,
$$\Pr\left[\max_{T \in \mathcal{T}} \left( \gain(F, T) - \Phi(T) \right) \ge \alpha d_\ell / \log^{0.25} m \right] \le 2^{-\Omega(\alpha)}.$$ 
\label{lem:gaintransitionpre}
\end{lemma}
\begin{proof}
For integer $r$ satisfying $|r| = O(\sqrt{\log m})$, let $T_r$ denote the $2^{\ell}$-dimensional vector that is $0$ in coordinates $[0, 2^{\ell} - 1] \setminus \{j\}$ and $r$ in coordinate $j$. Let $R$ be the set of $r$ for which $T_r$ is a valid $\ell$-transition vector that is compatible with $F$. 

For any $r \in R$, we have that $|\gain(F, T_r)|$ is given by $|a - b|$ where $a$ and $b$ are the number of blue and red points, respectively, in some geometric region with area $\Theta(rq_\ell d_\ell)$. Thus, by Lemma \ref{lem:fact2}, we have for any $\beta \ge r^{1.5}$ that, setting $\gamma = \beta / r^{1.5}$,
\begin{align*}
& \Pr[\gain(F, T_r) \ge  \beta \sqrt{q_\ell d_\ell}] \\
& = \Pr[\gain(F, T_r) \ge \gamma r \sqrt{r q_\ell d_\ell}] \\
& \le 2^{-\Omega(\gamma^2 r^2)} + 2^{-\Omega(\gamma r \sqrt{r q_\ell d_\ell})} \\
& \le 2^{-\Omega(\gamma^2 r^2)} + 2^{-\Omega(\gamma r^{1.5})} \tag{since $q_\ell d_\ell \ge 1$}\\
& \le 2^{-\Omega(\gamma r^{1.5})} \tag{since $\gamma \ge 1$}  \\
& = 2^{-\Omega(\beta)}.
\end{align*}

Thus
\begin{align*}
    & \Pr\left[\max_{r \in R} (\gain(F, T_r) - \Phi(T_r)) \ge \alpha d_\ell / \log^{0.25} m \right] \\
    & = \Pr\left[\max_{r \in R} (\gain(F, T_r) - \Phi(T_r)) \ge \alpha \sqrt{q_\ell d_\ell}\right] \\
    &  \le \sum_{r \in R} \Pr\left[\gain(F, T_r) \ge \Phi(T_r) + \alpha \sqrt{q_\ell d_\ell}\right] \\
    &  \le \sum_{r \in R} \Pr\left[\gain(F, T_r) \ge |r|^2 d_\ell / \log^{0.25} m + \alpha \sqrt{q_\ell d_\ell}\right] \\
    &  \le \sum_{r \in R} \Pr\left[\gain(F, T_r) \ge (|r|^2 + \alpha) \sqrt{q_\ell d_\ell}\right] \\
    & \le \sum_{r \in \mathbb{Z}} \exp(-\Omega(|r|^{2} + \alpha)) \\
    & \le 2^{-\Omega(\alpha)}.
\end{align*}

\end{proof}

Using vertical independence (Core Fact 1, stated in Lemma \ref{lem:fact1}), we can obtain a bound on $\max_T (\gain(F, T) - \Phi(T))$, where $T$ ranges over all $\ell$-transition vectors compatible with $F$.
\begin{lemma}
Consider a refined $\ell$-frame $F$, and let $\mathcal{T}$ be the set of $\ell$-transition vectors that are compatible with $F$. Then, for any $\alpha > 0$ that is at least a sufficiently large positive constant, we have
$$\Pr\left[\max_{T \in \mathcal{T}} \left( \gain(F, T) - \Phi(T) \right) \ge \alpha m / \log^{0.25} m\right] \le 2^{-\Omega(\alpha 2^\ell)}.$$ 
\label{lem:gaintransition}
\end{lemma}
\begin{proof}
For a given $\ell$-transition vector $T = \langle t_0, \ldots, t_{2^\ell - 1}\rangle$, let $T^{(i)}$ denote the $\ell$-transition vector that equals $0$ in coordinates $[0, 2^\ell - 1] \setminus \{i\}$ and that equals $t_i$ in coordinate $i$. Then, 
\begin{equation}
\max_{T \in \mathcal{T}} \left(\gain(F, T) - \Phi(T)\right)  = \sum_{j = 0}^{2^\ell - 1} \max_{T \in \mathcal{T}} \left(\gain(F, T^{(j)}) - \Phi(T^{(j)})\right).
\label{eq:diagsum}
\end{equation}   
Critically, the quantity $\gain(F, T^{(j)}) - \Phi(T^{(j)})$ is determined only by $t_j$ and by the blue/red dots in the vertical strip between where $F$ intersects diagonals $D^{(j)}$ and $D^{(j + 1)}$. It follows by Lemma \ref{lem:fact1} that \eqref{eq:diagsum} is a sum of $2^\ell$ independent random variables $X_j := \max_{T \in \mathcal{T}}(\gain(F, T^{(j)}) - \Phi(T^{(j)}))$, each of which by Lemma \ref{lem:gaintransitionpre} satisfies 
$$\Pr\left[X_j \ge \alpha d_\ell / \log^{0.25} m\right] \le 2^{-\Omega(\alpha)}.$$
Applying a Chernoff bound for independent geometric random variables, it follows that
\begin{align*}
 &   \Pr\left[\max_{T \in \mathcal{T}} \left( \gain(F, T) - \Phi(T) \right) \ge \alpha m / \log^{0.25} m\right] \\
 & = \Pr\left[\sum_{j = 0}^{2^\ell - 1} X_j \ge  \alpha m / \log^{0.25} m\right] \\
  & = \Pr\left[\sum_{j = 0}^{2^\ell - 1} X_j \ge  2^\ell \alpha d_\ell / \log^{0.25} m\right] \\
 & \le 2^{-\Omega(\alpha 2^\ell)}.
\end{align*}
\end{proof}

Recall that earlier in the section, we defined two types of gains: refinement gains and transitional gains. So far we have a bound on transitional gains (Lemma \ref{lem:gaintransition}). It is worth taking a moment to get an analogous bound on refinement gains (this bound is much easier to get, and does not require the use of the potential function $\Phi$).

\begin{lemma}
Consider an $\ell$-frame $F$ and let $\mathcal{F}$ denote the set of refined $\ell$-frames $F'$ that are compatible with $F$. For any $\alpha > 0$ that is at least a sufficiently large positive constant, we have
$$\Pr\left[\max_{F' \in \mathcal{F}} \gain(F, F') \ge \alpha m / \log^{0.25} m\right] \le 2^{-\Omega(\alpha 2^\ell)}.$$ 
\label{lem:gainrefinement}
\end{lemma}
\begin{proof}
Consider some refined $\ell$-frame $F' \in \mathcal{F}$. Let $P$ be the implied path for $F$ and $P'$ be the implied path for $F'$. Let $G^{+} \subseteq [m] \times [m]$ be the geometric region consisting of points that contained below $P'$ but not below $P$; and let $G^{-} \subseteq [m] \times [m]$ be the geometric region consisting of points that contained below $P$ but not below $P'$. Finally, let $B^+$ and $R^+$ (resp.~$B^-$ and $R^{-}$) be the number of blue and red points, respectively, in $G^+$ (resp.~$G^-$).

The quantity $\gain(F, F')$ is at most
$$|B^+ - R^+| + |B^- - R^-|.$$
Since the geometric regions $G^+$ and $G^-$ each have area at most $O(m q_\ell) = O\left(\frac{m^2}{2^\ell \sqrt{\log m}}\right)$, we have by Lemma \ref{lem:fact2} (i.e., Core Fact 2) that
\begin{align*}
    & \Pr\left[|B^+ - R^+| + |B^- - R^-| \ge \alpha m / \log^{0.25} m\right]\\
    & \le \exp\left(-\Omega\left(\frac{(\alpha m / \log^{0.25} m)^2}{\frac{m^2}{2^\ell \sqrt{\log m}}}\right)\right) + \exp(-\Omega(\alpha m / \log^{0.25} m))  \\
    & \le \exp\left(-\Omega\left(\alpha^2 2^\ell\right)\right) + \exp(-\Omega(\alpha m / \log^{0.25} m))  \\
    & \le \exp\left(-\Omega\left(\alpha 2^\ell\right)\right).  \tag{since $2^\ell \le O(m / \log^{0.5} m)$}\\ 
\end{align*}

Thus, for any fixed $F' \in  \mathcal{F}$, we have that $\Pr[\gain(F, F') \ge \alpha m / \log^{0.25} m] \le 2^{-\Omega(\alpha 2^\ell)}$. Applying a union bound,
\begin{align*}
    & \Pr\left[\max_{F' \in \mathcal{F}} \gain(F, F') \ge \alpha m / \log^{0.25} m\right] \\
    & \le  |\mathcal{F}| \cdot \exp\left(-\Omega\left(\alpha 2^\ell\right)\right) \\
    & \le 2^{2^\ell} \cdot \exp\left(-\Omega\left(\alpha 2^\ell\right)\right) \\
    & \le \exp\left(-\Omega\left(\alpha 2^\ell\right)\right). \tag{since $\alpha$ is sufficiently large}
\end{align*}
\end{proof}

Combining the previous lemmas, we can get a bound on the combined gains (both refinement and transitional) at a given level. As in Lemma \ref{lem:gaintransition}, the bound is parameterized by the potential function $\Phi$:
\begin{lemma}
Consider an $\ell$-frame $F$, and let $\mathcal{T}$ be the set of pairs $(F', F)$ such that $F'$ is a refined $\ell$-frame compatible with $F_{\ell}$ and $T$ is an $\ell$-transition vector compatible with $F'$. Then, for any $\alpha > 0$ that is at least a sufficiently large positive constant, we have
$$\Pr\left[\max_{(F', T) \in \mathcal{T}} \left( \gain(F, F') + \gain(F', T) - \Phi(T) \right) \ge \alpha m / \log^{0.25} m\right] \le 2^{-\Omega(\alpha 2^\ell)}.$$ 
\label{lem:gainfull}
\end{lemma}
\begin{proof}
We can bound 
\begin{align*}& \Pr\left[\max_{(F', T) \in \mathcal{T}} \left( \gain(F, F') + \gain(F', T) - \Phi(T) \right) \ge \alpha m / \log^{0.25} m\right] \\
&\le \Pr\left[\max_{(F', T) \in \mathcal{T}} \gain(F, F') \ge \frac{1}{2} \alpha m / \log^{0.25} m\right] \\
& \phantom{\text{foobar}} + \Pr\left[\max_{(F', T) \in \mathcal{T}} \left(\gain(F', T) - \Phi(T) \right) \ge \frac{1}{2} \alpha m / \log^{0.25} m\right]. \\
\end{align*}
The first probability is $2^{-\Omega(\alpha 2^\ell)}$ by Lemma \ref{lem:gainrefinement}. Defining $\mathcal{F}$ to be the set of $F'$ compatible with $F$, and $\mathcal{T}(F')$ to be the set of $T$ compatible with $F'$, the second probability can be expanded as
\begin{align*} 
& \Pr\left[\max_{F' \in \mathcal{F}} \max_{T \in \mathcal{T}(F')} \left(\gain(F', T) - \Phi(T) \right) \ge \frac{1}{2} \alpha m / \log^{0.25} m\right] \\
& \le \sum_{F' \in \mathcal{F}} \Pr\left[\max_{T \in \mathcal{T}(F')} \left(\gain(F', T) - \Phi(T) \right) \ge \frac{1}{2} \alpha m / \log^{0.25} m\right] \\
& \le \sum_{F' \in \mathcal{F}} 2^{-\Omega(\alpha 2^\ell)} \tag{by Lemma \ref{lem:gaintransition}}\\
& \le 2^{2^\ell} \cdot 2^{-\Omega(\alpha 2^\ell)}  \\
& \le 2^{-\Omega(\alpha 2^\ell)}. \tag{since $\alpha$ is at least a large constant}
\end{align*}

\end{proof}

For a given $\ell$-frame $F$, let us use $\mathcal{T}(F)$ to denote the set of pairs $(F', T)$, where $F'$ is a refined $\ell$-frame compatible with $F$ and $T$ is an $\ell$-transition vector compatible with $F'$ (and $F$). And let us use $\gain(F, F', T)$ as a shorthand for $\gain(F, F') + \gain(F', T)$.
So far, we have proven a bound of the form
\begin{equation}\Pr\left[\max_{(F', T) \in \mathcal{T}(F)} \left(\gain(F, F', T) - \Phi(T) \right) \ge \alpha m / \log^{0.25} m\right] \le 2^{-\Omega(\alpha 2^\ell)},
\label{eq:FTgain}
\end{equation}
where $F$ is a \emph{fixed} refined $\ell$-frame.

What we would really like, though, is a bound on 
\begin{equation}\max_{F \in \mathcal{F}_\ell} \max_{(F', T) \in \mathcal{T}(F)} \left(\gain(F, F', T) - \Phi(T) \right),\label{eq:dreamunion}\end{equation}
where $\mathcal{F}_\ell$ is the set of all possible $\ell$-frames. It is tempting to simply apply a union bound over $\mathcal{F}_\ell$. And, indeed, if we had $|\mathcal{F}_\ell| \le 2^{O(\ell)}$, then this union bound would be successful, giving the same dependency on $\alpha$ as in \eqref{eq:FTgain}. The problem is that $|\mathcal{F}_\ell|$ is actually $\sqrt{\log m}^{\Theta(2^\ell)} = 2^{\Theta(2^\ell \log \log m)}$. If we try to take a union bound over $F \in \mathcal{F}_\ell$, the bound that we will get is
\begin{equation}\Pr\left[\max_{F \in \mathcal{F}_\ell} \max_{(F', T) \in \mathcal{T}(F)}  \left( \gain(F, F', T) - \Phi(T) \right) \ge \alpha m / \log^{0.25} m\right] \le 2^{O(2^\ell \log \log m) -\Omega(\alpha 2^\ell)},
\label{eq:allFunion}
\end{equation}
which is only useful for $\alpha \ge \Omega(\log \log m)$.

The solution, it turns out, is to introduce a second potential function $\Psi$, this time determined by the $\ell$-frame $F$. Rather than bounding \eqref{eq:dreamunion}, we will actually bound
\begin{equation}\max_{F \in \mathcal{F}_\ell} \max_{(F', T) \in \mathcal{T}(F)} \left(\gain(F, F', T) - \Phi(T) 
 - \gamma \Psi(F)\right) \label{eq:realunion}\end{equation}
for some large positive constant $\gamma$.

To define $\Psi$, we must first define the \defn{compressed form} of a given $\ell$-frame (or refined $\ell$-frame) $F$. For a given $\ell \le \overline{\ell}$, and for a given $\ell$-frame $F = \langle f_0, f_1, \ldots, f_{2^\ell} \rangle$, define the compressed form of $F$ to be
$$\Delta(F) = (\Delta_1, \ldots, \Delta_{2^{\ell} - 1}),$$
where
$$\Delta_i = \frac{(f_{i + 1} - f_i) - (f_i - f_{i - 1})}{q_\ell}.$$
The way to think about the $\Delta_i$s is that they are the (discrete) second derivatives of the $f_i$s (normalized by a factor of $q_\ell$). We can visualize the $\Delta_i$s as in Figure \ref{fig:Deltas}. The reason that we call $\Delta(F)$ the \emph{compressed form} of $F$ is that we can always recover $F$ from $\Delta(F)$.\footnote{Indeed, given values for $(f_{i + 1} - f_i) - (f_i - f_{i - 1})$, with $i$ ranging from $1$ to $2^\ell - 1$, and given the equations $f_0 = 0$ and $f_{2^\ell} = 0$, we have $2^{\ell + 1}$ linearly independent equations for $2^{\ell + 1}$ variables $f_0, \ldots, f_{2^\ell}$.}\todo{Should we spell this out further?}

\begin{figure}
    \centering
    \includegraphics[scale=0.8]{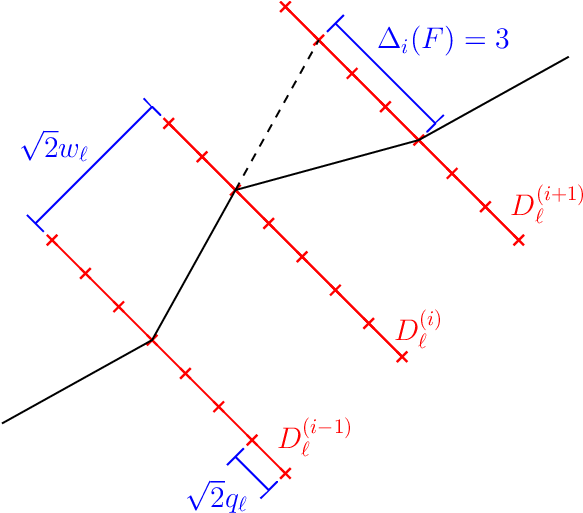}
    \caption{If we look at the implied path $\operatorname{Path}(F)$ for $F$, we take the line segment between diagonals $D^{(i - 1)}_\ell$ and $D^{(i)}_\ell$, and we extend that segment to reach diagonal $D^{(i + 1)}_\ell$, then $\Delta_i$ measures the distance (in multiples of $\sqrt{2} q_\ell$) between where the extended segment hits $D^{(i + 1)}_\ell$ versus where $\operatorname{Path}(F)$ hits $D^{(i + 1)}_\ell$. An example is shown in the figure, where the diagonals $D^{(i - 1)}_\ell, D^{(i)}_\ell, D^{(i + 1)}_\ell$ are in red with ticks every distance $q_\ell$; the path $\operatorname{Path}(F)$ is in black; the extension of the segment between $D^{(i - 1)}_\ell$ and $D^{(i)}_\ell$ is given as a dotted line; and $\Delta_i(F)$ is computed as $3$.}
    \label{fig:Deltas}
\end{figure}


We can now define the \defn{frame potential} $\Psi(F)$ of an $\ell$-frame (or refined $\ell$-frame) $F = \langle f_0, f_1, \ldots, f_{2^\ell}\rangle$ to be
$$\psi(F) = \sum_{0 < i < 2^\ell} |\Delta_i(F)| d_\ell / \log^{0.25} m.$$

The frame potential $\Psi(F)$ has two nice properties. Property 1 is that the sum of the potentials across the $\ell$-frames for a path $P$ is always at most $O(m \log^{0.75} m)$. This property is quite nontrivial and will be shown in Lemma \ref{lem:potential2}. Property 2 is that the number of bits needed to encode the compressed form of a given $\ell$-frame (or refined $\ell$-frame) is at most
$$O\left(2^\ell + \sum_i \log(1 + |\Delta_i(F)|)\right) \le O\left(2^\ell + \|\Delta_i(F)\|_1\right) \le O\left(2^\ell + \Psi(F) \cdot \frac{\log^{0.25} m}{d_\ell}\right).$$
Property 2 is a trivial consequence of the fact that $\Delta(F)$ encodes all information needed to recover $F$.

It is worth taking a moment to understand why these properties are useful. Property 1 tells us that it is okay to bound \eqref{eq:realunion} instead of \eqref{eq:dreamunion}, since for any given path $P$, we have $\sum_\ell \Psi(\frame'_\ell(P)) = O(m \log^{0.75} m)$. Property 2, on the other hand, tells us why we should be hopeful that a good bound on \eqref{eq:realunion} is possible. Since $|\mathcal{F}_\ell| = 2^{\Theta(2^\ell \log \log m)}$, the vast majority of frames $F \in \mathcal{F}_\ell$ must have compressed forms that require $\Omega(2^\ell \log \log m)$ bits to encode. By Property 2, these frames must also have large $\Psi$ potential (if we calculate it out, their potentials must be at least $\Omega(m \log \log m / \log^{0.25} m)$). In other words, for almost all $\ell$-frames $F$ and $(F', T) \in \mathcal{T}(F)$, the quantity $\gain(F, F', T) - \Phi(T) - \gamma \Psi(F)$ is actually much smaller than the quantity $\gain(F, F', T) - \Phi(T)$. This is why we will be able to bound \eqref{eq:realunion} even though we were not able to directly bound \eqref{eq:dreamunion}.

We now restate Property 1 as a lemma and prove it.
\begin{lemma}
Consider a path $P$, and for each $\ell \in [1, \overline{\ell}]$ let $F_\ell = \frame_\ell(P)$. Then,
\begin{equation}\sum_{\ell = 1}^{\overline{\ell}} \psi(F_\ell)  = O(m \log^{0.75} m).
\label{eq:Psisum1}
\end{equation}
\label{lem:potential2}
\end{lemma}
\begin{proof}

For $\ell < \overline{\ell}$, let $F'_\ell = \frame'_\ell(P)$ and define $E_{\ell + 1}$ to be the $(\ell + 1)$-frame whose implied path is the same as the implied path for $F'_\ell$. Finally, let $T_\ell = \langle t_0, t_2, \ldots, t_{2^{\ell} - 1}\rangle = \trans_\ell(P)$. We will show that, for $\ell < \overline{\ell}$:
\begin{equation}
    \|\Delta(F'_\ell)\|_1 \le \|\Delta(F_\ell)\|_1 + O(2^\ell),
    \label{eq:psitransform1}
\end{equation}
\begin{equation}
    \|\Delta(E_{\ell + 1})\|_1 = \|\Delta(F'_\ell)\|_1,
    \label{eq:psitransform2}
\end{equation}
\begin{equation}
    \|\Delta(F_{\ell + 1})\|_1 \le \|\Delta(E_{\ell + 1})\|_1 + O\left(2^\ell + \sum_i t_i\right).
    \label{eq:psitransform3}
\end{equation}

Equations \eqref{eq:psitransform1} and \eqref{eq:psitransform3} can be proven together, as they are both special cases of the following identity: for any two $\ell$-frames $A$ and $B$,
\begin{equation}
    \|\Delta(A)\|_1 \le \|\Delta(B)\|_1 + O\left(\frac{\|A - B\|_1}{q_\ell}\right).
    \label{eq:deltadelta}
\end{equation}
Indeed, if we expand out $\|\Delta(A)\|_1 - \|\Delta(B)\|_1$, we get 
\begin{align*}
\|\Delta(A)\|_1 - \|\Delta(B)\|_1 & =  \sum_i \frac{|(A_{i + 1} - A_i) - (A_i - A_{i - 1})| - |(B_{i + 1} - B_i) - (B_i - B_{i - 1})|}{q_\ell} \\
& \le \sum_i \frac{|A_{i + 1} - B_{i + 1}| + |A_i - B_i| + |A_i - B_i| + |A_{i - 1} - B_{i - 1}|}{q_\ell} \\
& \le \frac{4\|A - B\|_1}{q_\ell}.
\end{align*}
This proves \eqref{eq:deltadelta}, which then implies \eqref{eq:psitransform1} and \eqref{eq:psitransform2}.

Now, to prove \eqref{eq:psitransform2}, let $e_i$ denote the $i$-th coordinate of $E_{\ell + 1}$ and let $f
_i$ denote the $i$-th coordinate of $F_\ell'$. Observe that for even coordinates $2i$, we have $e_{2i} = f'_i$, and that for odd coordinates $2i + 1$, we have $e_{2i + 1} = \frac{1}{2} f'_i + \frac{1}{2} f'_{i + 1}$. It follows that
\begin{align*}
    \Delta_{2i}(E_{\ell + 1}) & = \frac{e_{2i + 1} + e_{2i - 1} - 2e_{2i}}{q_{\ell + 1}}\\
         & = \frac{\frac{1}{2} f'_{i} + \frac{1}{2} f'_{i + 1} + \frac{1}{2} f'_{i - 1} + \frac{1}{2} f'_{i} - 2f_{i}}{q_{\ell + 1}}\\
         & = \frac{\frac{1}{2} f'_{i + 1} + \frac{1}{2} f'_{i - 1} - f'_{i}}{q_{\ell + 1}}\\
         & = \frac{f'_{i + 1} + f'_{i - 1} - 2f'_{i}}{q_{\ell}}\\
         & = \Delta_{i}(F'_{\ell}),
\end{align*}
and that
\begin{align*}
    \Delta_{2i + 1}(E_{\ell + 1}) & = \frac{e_{2i + 2} + e_{2i} - 2e_{2i + 1}}{q_{\ell + 1}}\\
         & = \frac{f'_{i + 1} + f'_{i} - 2 \cdot \frac{1}{2}  (f'_i + f'_{i + 1})}{q_{\ell + 1}}\\
         & = 0,
\end{align*}
which together imply \eqref{eq:psitransform2}.

Combining \eqref{eq:psitransform1}, \eqref{eq:psitransform2}, and \eqref{eq:psitransform3}, we have
\begin{align*}
   \Psi(F_{\ell + 1}) & = \| \Delta(F_{\ell + 1})\|_1 d_{\ell + 1} / \log^{0.25}m \\
    & \le \| \Delta(F_{\ell})\|_1 d_{\ell +1}/ \log^{0.25}m + O(2^\ell d_{\ell +1} / \log^{0.25} m) + O\left(\sum_i t_i d_{\ell +1} / \log^{0.25} m \right) \\
        & \le \frac{1}{2} \| \Delta(F_{\ell})\|_1 d_\ell / \log^{0.25}m + O(2^\ell d_\ell / \log^{0.25} m) + O\left(\sum_i t_i d_\ell / \log^{0.25} m \right) \\
    & \le \frac{1}{2} \Psi(F_{\ell}) + O\left(m / \log^{0.25} + \Phi(T_\ell)\right). \\
\end{align*}

Since $\Psi(F_0) = 0$, it follows that
\begin{align*}
    \sum_{\ell = 0}^{\overline{\ell}}  \Psi(F_{\ell}) & \le \sum_{\ell = 1}^{\overline{\ell}} O\left(m / \log^{0.25} m +\Phi(T_\ell)\right) \cdot \sum_{j \ge 0} \frac{1}{2^j} \\
    & = \sum_{\ell = 1}^{\overline{\ell}} O\left(m / \log^{0.25} m+ \Phi(T_\ell)\right) \\
    & \le O(m \log^{0.25} m). \tag{by Lemma \ref{lem:potential1}} \\
\end{align*}
This establishes \eqref{eq:Psisum1}, as desired.
\end{proof}

Next, we use Property 2 to obtain a bound on \eqref{eq:realunion}.
\begin{lemma}
Let $\ell < \overline{\ell}$, and let $\mathcal{F}$ be the set of all triples $(F, F', T)$, where $F$ is an $\ell$-frame, $F'$ is a refined $\ell$-frame compatible with $F$, and $T$ is an $\ell$-transition vector compatible with $F'$. Then, there exists a positive constant $\gamma$ such that for any $\alpha$ that is at least a sufficiently large positive constant, we have
$$\Pr\left[\max_{(F, F', T) \in \mathcal{F}} \left(\gain(F, F', T) - \Phi(T) - \gamma \psi(F)\right) \ge \alpha m / \log^{0.25} m\right] \le 2^{-\Omega(\alpha 2^\ell)}.$$ 
\label{lem:maxgainbound}
\end{lemma}
\begin{proof}
Fix $\ell < \overline{\ell}$, and for $q \ge 0$ let $\mathcal{C}_{q}$ be the set of $\ell$-frames $F$ satisfying $\|\Delta(F)\|_1 = q$. Recall Property 2 from earlier: that, since $\Delta(F)$ encodes $F$, and since we can write $\Delta(F)$ in unary using
$$O(|\Delta(F)\|_1 + 2^\ell)$$
bits, it follows that we can encode $F$ itself in $O(\|\Delta(F)\|_1 + 2^\ell)$ bits. We therefore have that
\begin{equation}|\mathcal{C}_q| \le 2^{O(q) + \ell}.
\label{eq:Csize}
\end{equation}

For a given $\ell$-frame, define $\mathcal{T}(F)$ to be the set of pairs $(F', T)$ where $F'$ is a refined $\ell$-frame, $T$ is an $\ell$-transition vector, and $F, F', T$ are compatible with each other. Then we can bound
\begin{align*}
&\Pr\left[\max_{(F, F', T) \in \mathcal{F}} \left(\gain(F, F', T) - \Phi(T) - \gamma \psi(F)\right) \ge \alpha m / \log^{0.25} m\right] \\
& \le \sum_{q \ge 0} \sum_{F \in \mathcal{C}_q} \Pr\left[\max_{(F', T) \in \mathcal{T}(F)} \left(\gain(F, F', T) - \Phi(T) - \gamma \psi(F)\right) \ge \alpha m / \log^{0.25} m \right] \\
& \le \sum_{q \ge 0} \sum_{F \in \mathcal{C}_q} \Pr\left[\max_{(F', T) \in \mathcal{T}(F)} \left(\gain(F, F', T) - \Phi(T) - \gamma \|\Delta(F)\|_1 d_\ell / \log^{0.25} m\right) \ge \alpha m / \log^{0.25} m \right] \\
& \le \sum_{q \ge 0} \sum_{F \in \mathcal{C}_q} \Pr\left[\max_{(F', T) \in \mathcal{T}(F)} \left(\gain(F, F', T) - \Phi(T) - \gamma \frac{\|\Delta(F)\|_1}{2^\ell} m / \log^{0.25} m\right) \ge \alpha m / \log^{0.25} m \right] \\
& =  \sum_{q \ge 0} \sum_{F \in \mathcal{C}_q} \Pr\left[\max_{(F', T) \in \mathcal{T}(F)} \left(\gain(F, F', T) - \Phi(T)\right) \ge (\alpha m + \gamma q / 2^\ell) / \log^{0.25} m \right] \\
& \le \sum_{q \ge 0} \sum_{F \in \mathcal{C}_q} \exp(-\Omega(\alpha 2^\ell + \gamma q)) \tag{by Lemma \ref{lem:gainfull}} \\
& \le \sum_{q \ge 0} 2^{O(q) + \ell} \exp(-\Omega(\alpha 2^\ell + \gamma q)) \tag{by \eqref{eq:Csize}} \\
& \le \sum_{q \ge 0} \exp(-\Omega(\alpha 2^\ell + \gamma q)) \tag{since $\alpha, \gamma$ at least sufficiently large constants} \\
& \le \exp(-\Omega(\alpha 2^\ell)).
\end{align*}
\end{proof}


At this point, we are nearly ready to prove the main result of the section, that is, that the expected maximum surplus of any path is $O(m \log^{0.75} m)$. We just need two more lemmas. The first gives us a baseline for the surplus of the trivial straight-line path, and the second bounds the maximum possible difference between $\surplus(\frame_{\overline{\ell}}(P))$ and $\surplus(P)$ for any path. \todo{$(x, y)$ vs $x$}

\begin{lemma}
    Let $P_0$ be the straight-line path from $(0, 0)$ to $(m, m)$. Then, for $\alpha \ge 1$,
    $$\Pr\left[\surplus(P_0) \ge \alpha m \sqrt{\log m} \right] \le e^{-\Omega(\alpha \log m)}.$$
    \label{cor:P0}
\end{lemma}
\begin{proof}
    By Lemma \ref{lem:fact2}, with $A = m^2 / 2$, 
    \begin{align*}
        \Pr\left[\surplus(P_0) \ge \alpha m \sqrt{\log m} \right] & \le 2^{-\Omega(\alpha^2 \log m)} + 2^{-\Omega(\alpha m \sqrt{\log m})} \\
        & \le e^{-\Omega(\alpha \log m)}.
    \end{align*}
\end{proof}

\begin{lemma}
Let $\mathcal{P}$ be the set of all paths. For any $\alpha > 0$ that is at least a sufficiently large positive constant,
$$\Pr\left[\max_{P \in \mathcal{P}} |\surplus(\frame_{\overline{\ell}}(P)) - \surplus(P)| \ge \alpha m \log^{0.75} m\right] \le \exp\left(-\Omega\left(\alpha m \log^{0.75} m\right)\right).$$
\label{lem:finalrefinement}
\end{lemma}
\begin{proof}
Let $F$ be an $\overline{\ell}$-frame, let $Q$ be the implied path for $F$, and let $G_F \subseteq [m] \times [m]$ be the geometric region consisting of all $(x, y)$ that have Euclidean distance at most $4d_{\overline{\ell}} = O(\sqrt{\log m})$ to the nearest point on $Q$. If a path $P$ satisfies $\frame_{\overline{\ell}}(P) = F$, then the path $P$ is necessarily contained in the region $G_F$. Since $G_F$ has total area $O(m d_\ell) = O(m \sqrt{\log m})$, it follows from Lemma \ref{lem:fact3} (i.e., Core Fact 3) that
\begin{equation}
\Pr\left[ \max_{P \mid \frame_{\overline{\ell}}(P) = F} |\surplus(F) - \surplus(P)| \ge \alpha m \log^{0.75} m \right] \le 2^{-\Omega(\alpha m \log^{0.75} m)}.
    \label{eq:chernoff_second}
\end{equation}
Defining $\mathcal{F}$ to be the set of all $\overline{\ell}$-frames, we can take a union bound to get
\begin{align*}
&    \Pr \left[\max_{P \in \mathcal{P}} |\surplus(\frame_{\overline{\ell}}(P)) - \surplus(P)| \ge \alpha m \log^{0.75}  m \right] \\
& \le |\mathcal{F}|  \cdot 2^{-\Omega(\alpha m \log^{0.75}  m)} \\
& \le 2^{o(m) -\Omega(\alpha m \log^{0.75}  m)} \\
& \le 2^{-\Omega(\alpha m \log^{0.75}  m)},
\end{align*}
as desired.
\end{proof}

Finally, we can prove the main result of the section.
\begin{theorem}
Let $\mathcal{P}$ be the set of all paths. Then, for any $\alpha > 0$ that is at least a sufficiently large positive constant,
$$\Pr\left[\max_{P \in \mathcal{P}} \surplus(P) \ge \alpha m \log^{0.75} m \right] \le 2^{-\Omega(\alpha \log m)}.$$
\label{thm:surplusupper}
\end{theorem}
\begin{proof}
Let $\alpha > 1$, and let $\alpha_1, \alpha_2, \ldots, \alpha_{\overline{\ell}} \ge 0$ be values that we will select later, and that have average at most $\alpha$. Let $\mathcal{F}_\ell$ be the set of triples $(F, F', T)$ where $F$ is an $\ell$-frame, where $F'$ is a refined $\ell$-frame compatible with $F$, and where $T$ is an $\ell$-transition vector compatible with $F'$. 

Let $c_1$ and $c_2$ be sufficiently large positive constants, where $c_1$ is sufficiently large as a function of $c_2$ (so $c_1 \gg c_2$). By Lemmas \ref{lem:potential1} and \ref{lem:potential2}, we have that
\begin{align}
   & \Pr\left[\max_{P \in \mathcal{P}} \surplus(P) \ge (\alpha + c_1)m \log^{0.75} m \right] \\
   & \le \Pr\left[\max_{P \in \mathcal{P}} \left( \surplus(P) - \sum_{\ell = 0}^{\overline{\ell} - 1} c_2 \Psi(\frame_\ell(P)) + \Phi(\trans_\ell(P))\right)  \ge \alpha m \log^{0.75} m\right].\label{eq:surplusprob}
   \end{align}
   Define $P_0$ to be the straight-line path from $(0, 0)$ to $(m, m)$. Then, we can expand $\surplus(P)$ as
\begin{align*}& \surplus(P_0) + \max_{P \in \mathcal{P}} \sum_{\ell = 0}^{\overline{\ell} - 1}\left(\surplus(\frame_{\ell + 1}(P)) - \surplus(\frame_\ell(P))\right)  + \surplus(P) - \frame_{\overline{\ell}}(P). 
\end{align*}
We can therefore upper bound \eqref{eq:surplusprob} by 
\begin{align*}
   & \Pr\left[\surplus(P_0) \ge  \frac{\alpha}{3} m \log^{0.75 m}\right] + \\
   & + \Pr\Bigg[\max_{P \in \mathcal{P}} \sum_{\ell = 0}^{\overline{\ell} - 1}\left(\surplus(\frame_{\ell + 1}(P)) - \surplus(\frame_\ell(P)) - c_2 \Psi(\frame_\ell(P)) - \Phi(\trans_\ell(P))\right)  \\ & \phantom{foobargoobar} \ge \frac{\alpha}{3} m \log^{0.75} m  \Bigg]  \\
   & + \Pr\left[ \surplus(P) - \frame_{\overline{\ell}}(P) \ge \frac{\alpha}{3} m \log^{0.75 m} \right]. \\
\end{align*}
By Corollary \ref{cor:P0} and Lemma \ref{lem:finalrefinement}, the first and final probabilities are each $e^{-\Omega(\alpha \log m)}$. Thus, we can focus the rest of the proof on the middle probability. We can upper bound the probability by
\begin{align*}
& \Pr\Bigg[\sum_{\ell = 0}^{\overline{\ell} - 1} \max_{(F, F', T) \in \mathcal{F}_\ell} \left(\gain(F, F', T) - c_2 \Psi(F) - \Phi(T)\right) \ge \frac{\alpha}{3} m \log^{0.75} m  \Bigg].  \\
\end{align*}

Using the fact that $\operatorname{avg}(\alpha_i) \le \alpha$, we can further upper bound the probability by
   \begin{align*}
& \le \sum_{\ell = 0}^{\overline{\ell} - 1} \Pr\left[\max_{(F, F', T) \in \mathcal{F}_\ell} \left(\gain(F, F', T) - c_2 \Psi(F) - \Phi(T) \right) > \frac{\alpha_\ell}{3} m / \log^{0.25} m\right].
   \end{align*}
   By Lemma \ref{lem:maxgainbound}, this is at most
   \begin{align*}       
   & \le \sum_{\ell = 0}^{\overline{\ell} - 1}  \exp(-\Omega(\alpha_\ell 2^\ell)). \\
\end{align*}

Now set $\alpha_\ell = \frac{1}{2}\alpha + \frac{\alpha \log m}{2^{\ell + 3}}$. These are valid $\alpha_\ell$s since their average satisfies
\begin{align*}
\frac{1}{\overline{\ell}} \sum_{\ell = 0}^{\overline{\ell}} \alpha_\ell & \le \frac{1}{2} \alpha + \frac{1}{\overline{\ell}} \sum_{\ell = 0}^{\overline{\ell}} \frac{\alpha \log m}{2^{\ell + 3}} \\
& \le  \frac{1}{2} \alpha + \frac{2}{\log m} \sum_{\ell \ge 0} \frac{\alpha \log m}{2^{\ell + 3}} \\
& \le  \frac{1}{2}\alpha + \frac{1}{2} \alpha\\
& \le \alpha. \\
\end{align*}

Plugging the $\alpha_\ell$s in, we get 
\begin{align*}
   & \sum_{\ell = 0}^{\overline{\ell} - 1} \exp(-\Omega(\alpha_\ell 2^\ell)) \\
   & \le  \sum_{\ell = 0}^{\overline{\ell} - 1}  \exp(-\Omega(\alpha 2^\ell + \alpha \log m)) \\
   & \le \sum_{\ell = 0}^{\overline{\ell} - 1} \exp(-\Omega(\alpha 2^\ell + \alpha \log m)) \\
   & \le \exp(-\Omega(\alpha \log m)),
\end{align*}
which completes the proof.
\end{proof}

\begin{corollary}
Let $\mathcal{P}$ be the set of all paths. Then,
    $$\E\left[\max_{P \in \mathcal{P}} \surplus(P)\right] \le O(m \log^{0.75} m).$$
\end{corollary}

\begin{corollary}
Let $\mathcal{P}$ be the set of all paths. Let $S = \max_{P \in \mathcal{P}} \surplus(P)$. For any positive constant $c_1$, there exists a positive constant $c_2$ such that
    $$\E\left[\max(0, S - c_2 m \log^{0.75} m)\right] \le \frac{1}{m^{c_1}}.$$
\end{corollary}

\section{From Path Surplus to Insertion Surplus}\label{sec:insertionsurplus}

In this section, we extend our analysis of the Path Surplus Problem (Section \ref{sec:pathsurplus}) in order to get tight bounds for the so-called Insertion Surplus Problem~\cite{bender2022linearfull}. Then, in Section \ref{sec:linprobing}, we will show how to use these bounds in order to perform a tight analysis of the amortized expected complexity of ordered linear-probing hash tables.

Let $U$ be a universe, let $n$ be a parameter, and let $R = O(n)$ be even. Consider a sequence of $R$ operations $o_1, o_2, \ldots, o_{2R}$, alternating between deletions and insertions. The only restriction on the deletions and insertions is that each time an element $x \in U$ is inserted (resp.~deleted), it must subsequently be deleted (resp.~inserted) before it can again be inserted (resp.~deleted).

Define the \defn{Insertion-Surplus Problem} as follows. Let $h$ be a fully random hash function $U \rightarrow[n]$, and let $t \le n$ (one should think of $t = n^{o(1)}$). Define $\mu = \frac{t}{n} \cdot R$. Plot blue/red dots as follows:
\begin{itemize}
\item a blue dot at $(h(x_i), i)$ for each insertion $o_i$ of an element $x_i$;
\item a red dot at $(h(x_i), i)$ for each deletion $o_i$ of an element $x_i$.
\end{itemize}
Now consider the $t \times R$ grid containing dots $(a, b)$ satisfying $a \in [0, t]$ and $b \in [1, R]$. Consider the monotonic paths through the grid, going from $(0, 1)$ to $(t, R + 1)$, and define the surplus of each path, as in previous sections, to be the number of blue dots minus the number of red dots (strictly) beneath the path. Note the only dots that can be beneath the path are those corresponding to insertions/deletions with hashes in $[1, t]$.

The total expected number of blue (or red) dots in the grid is given by $\mu$, so one should think of $\mu$ as being a proxy for $m^2$ in the original Path Surplus Problem (Section \ref{sec:pathsurplus}). We will prove that, for any choice of insertions/deletions $R$, the expected maximum surplus of any path is $O(\sqrt{\mu} \log^{0.75} \mu + t^2 / n)$ (Section \ref{sec:insertionsurplusupper}); and that, in the case where every operation $o_i \in R$ is on a different key, there is also a lower bound of $\Omega(\sqrt{\mu} \log^{0.75 m} - t^2 / n)$ (Section \ref{sec:insertionsurpluslower}).

\subsection{Upper Bound on Insertion Surplus}\label{sec:insertionsurplusupper}

In this subsection, we prove an upper bound of $O(\sqrt{\mu} \log^{0.75} \mu + t^2/n)$ on the expected maximum surplus of any path. To simplify discussion, we assume throughout the subsection that the first operation $o_1$ is a deletion, and the final operation $o_{2R}$ is an insertion. So odd-indexed operations are deletions and even-indexed operations are insertions. (This will matter for parity edge cases in some of our arguments, but will not matter for the final result.)

To make this problem more closely resemble the Path Surplus problem from Section \ref{sec:pathsurplus}, we can make a sequence of three modifications to the problem that have provably negligible (or decreasing) effect on the maximum path surplus:

\paragraph{Modification 1: Poissonization.} The first step is to (slightly) modify the process for generating dots. Rather than using a single hash function $h$, we will now use an infinite sequence $h_1, h_2, \ldots$, and rather than placing a single dot for operation $o_i$, we will generate dots as follows: For each $x \in U$, generate a Poisson random variable $Q_x \sim \Pois(1)$; and for each operation $o_i$ that inserts/deletes $x$ and for each $j \in \{1, 2, \ldots, Q_x\}$, place a blue/red dot (depending if $o_i$ is an insertion/deletion, respectively) at $(h_j(x_i), i)$. In other words, rather than placing $1$ dot per key $x$, we have associated $x$ with a Poisson-random-variable $Q_x$ number of dots, each with different hashes.

This type of Poissonization is a standard trick (see, e.g., discussion in Section 5.3 of \cite{mitzenmacherbook}) for creating independence between bins in balls-to-bins settings. In this case, if we define bins $1, 2, \ldots$ so that bin $i$ consists of the dots in the $i$-th column (i.e., of the form $(i, \cdot)$), then Poissonization guarantees that the bins are independent random variables.

A bit of care is needed to bound the impact of this modification on the maximum path surplus. Here, again, we can use a standard approach for how to think about Poissonization \cite{mitzenmacherbook}. It is possible to create a coupling between the un-Poissonized and the Poissonized versions of the process such that the blue/red dots that land in bins $1, 2, \ldots, t$ are precisely the same in both processes, except for an $O(t / n)$-expected fraction of the dots.\footnote{We can construct the coupling as follows, where for convenience we call a dot \defn{corrupt} if it appears one version of the process but not the other. Call a dot in the Poissonized process \defn{duplicate corrupted} if it is one of at least two dots in $[1, t] \times [1, R + 1]$ that are created by the same operation $o_i$ as each other. The expected fraction of dots that are duplicate corrupted is at most  $O(t / n)$. On the other hand, the probability of placing exactly one dot for $o_i$ in $[1, t] \times [1, R + 1]$ is $\frac{t}{n}$ for the non-Poissonized process and $\frac{t}{n} - O(t^2 / n^2)$ for the Poissonized process (call these \defn{singleton dots}). Thus, by labeling an $O(t / n)$-expected-fraction of the singleton dots in the non-Poissonized process as \defn{singleton corrupt}, we can couple the remaining singleton dots to be the same in both processes. In total, we have labeled at most an $O(t/n)$-expected fraction of the dots that land in $[1, t] \times [1, R + 1]$ as corrupt in each process, and coupled the remaining dots to be the same in both processes, as desired.} It follows that the expected difference in maximum path surplus is at most $O(\mu \cdot t / n) = O(t^2 / n)$.

\paragraph{Modification 2: Smoothing the dot placement.} Let us modify the hash functions $h_1, h_2, \ldots$ to output random real numbers in $(0, n]$ (so if previously $h_i$ produced some integer $\ell$, now it produces a random real number in $(\ell - 1, \ell]$). Moreover, rather than placing the $j$-th blue/red dot for $o_i$ directly at $(h_j(x_i), i)$, we will place it as follows:
\begin{itemize}
    \item If $j$ is odd (i.e., $o_j$ is a deletion), place a red dot at $(h_j(x_i), b)$ where $b$ is uniformly random in $(j, j + 2]$;
    \item If $j$ is even (i.e., $o_j$ is an insertion), place a blue dot at  $(h_j(x_i), b)$ where $b$ is uniformly random in $(j - 1, j + 1]$;
\end{itemize}

\begin{lemma}
    Smoothing the dot placement does not decrease the maximum path surplus.
    \label{lem:smooth1}
\end{lemma}
\begin{proof}
    Prior to smoothing the dot placement, we had without loss of generality that every path traveled along integer grid lines. Let $P$ be such a path.
    
    To analyze the effect of smoothing the dot placement on $P$, let us break the smoothing process into two phases. First, suppose that we perform the following partial smoothing process (the difference is in how we pick $b$):
    \begin{itemize}
    \item If $j$ is odd (i.e., $o_j$ is a deletion), place a red dot at $(h_j(x_i), b)$ where $b$ is uniformly random in $[j, j + 1)$;
    \item If $j$ is even (i.e., $o_j$ is an insertion), place a blue dot at   $(h_j(x_i), b)$ where $b$ is uniformly random in $[j, j + 1)$;
    \end{itemize}
    In this version of smoothing each dot moves from an integer point $(x, y)$ to a real-valued point in $(x - 1, x] \times [y, y + 1)$. This does not change $\surplus(P)$ at all, since the set of points beneath $P$ is precisely the same as before.

    Now, to get from the partial smoothing process to the full smoothing process, we can perform the following additional modification:
    \begin{itemize}
        \item Each red dot that is currently in some position $(x, y)$ gets moved to $(x, y + 1)$ with probability $1/2$.
        \item Each blue dot that is currently in some position $(x, y)$ gets moved to $(x, y - 1)$ with probability $1/2$.
    \end{itemize}
    This can only increase $\surplus(P)$ since we are moving red dots up and blue dots down. On the other hand, when we combine the partial smoothing process with this random movement step, we get the full smoothing process, which completes the proof of the lemma.
\end{proof}

The purpose of smoothing is to ensure that the probability density of blue/red dots is uniform throughout $[0, t] \times [1, R + 1]$.
\begin{lemma}
    Let $a \in \{1, 2, \ldots, t\}$ and $b \in \{1, 2, \ldots, R\}$. Consider the geometric region $G = (a - 1, a] \times [b, b + 1)$. Then the number of blue dots $B$  and the number of red dots $R$ in $G$ are both Poisson random variables with mean $1/(2n)$, and the dots in $G$ are placed at uniformly random (and mutually independent) positions in $G$.
\end{lemma}
\begin{proof}
    Suppose $b$ is even. Prior to smoothing, there were $\Pois(1/n)$ blue points (resp.~red points) at $(a, b)$ (resp.~$(a, b - 1)$), each of which now has a 50\% probability of being placed uniformly at random in $G$ (due to smoothing). 

    Suppose $b$ is odd. Prior to smoothing, there were $\Pois(1/n)$ blue points (resp.~red points) at $(a, b + 1)$ (resp.~$(a, b)$), each of which now has a 50\% probability of being placed uniformly at random in $G$ (due to smoothing). 
\end{proof}

As an immediate corollary, we get:
\begin{lemma}
    Suppose we rescale the $[0, t] \times [1, R + 1]$ grid to be on $[m] \times [m]$, where $m = \sqrt{\mu}$. Consider any geometric region $G$, and let $A$ be the area of $G$. Let $B$ and $R$ be the number of blue and red dots in $G$, respectively. Then,
    $$\E[B] = \E[R] = A.$$
    \label{lem:geometricexpectation}
\end{lemma}

It is worth noting that, in this new setting, Poissonization gives us full vertical independence: if we partition the $[0, t] \times [1, R + 1]$ grid into disjoint vertical strips $V_1, V_2, \ldots$, then the sets of blue/red dots in each vertical strip are mutually independent. We will spell this out in more detail in Lemma \ref{lem:vertind}.

\paragraph{Modification 3: Rescaling the grid.} Finally, we can rescale the $[0, t] \times [1, R + 1]$ grid to be on $[m] \times [m]$ where $m = \sqrt{\mu}$. This, of course, has no effect on the maximum path surplus but has the convenient effect that it plays well with Lemma \ref{lem:geometricexpectation}.

\paragraph{Analysis of the modified problem.} To establish an upper bound on the maximum path surplus in this (modified) setting, we must re-establish the three core facts from Section \ref{sec:pathsurplusupper} (Lemmas \ref{lem:fact1}, \ref{lem:fact2}, and \ref{lem:fact3}). 

For $x \in U$, say that there is a \defn{$x$-column} at real-valued position $r \in [0, n]$ if $r \in \{h_1(x), h_2(x), \ldots, h_{Q_x}(x)\}$. We think of the $x$-column as consisting of the blue/red dots that are placed at coordinates of the form $(r, \cdot)$ by insertions/deletions of $x$. The key property that Poissonization offers us is that, for a given $x$, the positions $r$ that contain $x$-columns are generated by a Poisson process. This is what gives us vertical independence:

\begin{lemma}[Core Fact 1: Vertical Independence]
Partition the $[m] \times [m]$ grid into disjoint vertical strips $V_1, V_2, \ldots$, and define $B_i$ and $R_i$ to be the sets of blue and red dots in strip $V_i$, respectively. Then the pairs $(B_1, R_1), (B_2, R_2), \ldots$ are mutually independent random variables.
\label{lem:vertind}
\end{lemma}
\begin{proof}
    For each item $x \in U$, we hash $x$ to a Poisson $Q_x \sim \Pois(1)$ number of real-valued points $r_1, r_2, \ldots, r_{Q_x} \in [0, n]$ (the points are real-valued due to the smoothing modification) and place an $x$-column at each horizontal coordinate $r_i$. The set of positions $r_1, r_2, \ldots$ can therefore be viewed as being generated by a Poisson process. It follows that, for each $V_i$, the number $C_{i, x}$ of $x$-columns that $V_i$ receives from $x$ is a Poisson random variable that is independent of the other $C_{j, y}$s (where $(j, y) \neq (i, x)$). This, in turn, implies that the sets of blue/red dots in each vertical strip are mutually independent.
\end{proof}

\begin{lemma}[Core Fact 2: Region Surplus]
Consider any geometric region $G$, and let $A$ be the area of $G$. Let $B$ (resp.~$R$) denote the number of blue dots (resp.~red dots) that appear in $G$. Then,
\begin{equation}\Pr[|B - R| \ge \kappa] \le e^{-\Omega(\kappa^2 / A)} + e^{-\Omega(\kappa)}.\end{equation}
\label{lem:chernoff_first_repeat}
\end{lemma}
\begin{proof}
Call an $x$-column \defn{positive} if it increases $B - R$, \defn{negative} if it decreases $B - R$, and \defn{neutral} otherwise. A given $x$-column can affect $B - R$ by at most $1$. If we define $C^+$ to be the number of positive columns and $C^-$ to be the number of negative ones, then
$$B - R = C^+ - C^-.$$
We know from Lemma \ref{lem:geometricexpectation} that $\E[B]$ and $\E[R]$ are equal, so $\E[B - R] = \E[C^+ - C^-] = 0$. Thus, we can complete the proof by proving individual concentration bounds for each of $C^+$ and $C^-$ individually. Since our argument will be the same for both quantities, let us focus on $C^+$. 

We begin by bounding $\E[C^+]$. We can upper-bound $C^+$ by the number $C$ of $x$-columns (considering all $x \in U$) that place at least one blue dot in $G$. We know from Lemma \ref{lem:geometricexpectation} that $\E[C] \le A$. Thus $\E[C^+] \le A$.

Now, to prove a concentration bound on $C^+$, observe that we can express $C^+$ as a sum of independent Poisson random variables (one for each $x \in U$). This means that $C^+$ itself is a Poisson random variable with mean at most $A$. It follows that
$$\Pr[C^+ - \E[C^+] \ge \kappa] \le e^{-\Omega(\kappa^2 / A)} + e^{-\Omega(\kappa)}.$$
The same argument gives an analogous bound for $C^-$, thereby completing the proof.
\end{proof}

\begin{lemma}[Core Fact 3: Restricted-Path Upper Bound]
Consider any geometric region $G$, and let $A$ be the area of $G$. Let $\mathcal{P}_G$ be the set of monotonic paths that stay within $G$ at all times. Define
$$S = \max_{P \in \mathcal{P}_G} \surplus(P) - \min_{P \in \mathcal{P}_G} \surplus(P).$$
For $\alpha \ge 1$,
\begin{equation}\Pr[S \ge \alpha A] \le e^{-\Omega(\alpha) \cdot A}.
\label{eq:SAbound}
\end{equation}
\end{lemma}
\begin{proof}
    We can upper-bound $S$ by the number $C$ of $x$-columns (considering all $x \in U$) that place at least one blue or red dot in $G$. We know from Lemma \ref{lem:geometricexpectation} that $\E[C] \le 2A$. Since, furthermore, $C$ is a sum of independent Poisson-random variables (one for each key $x$), we have that $C$ is itself a Poisson random variable. This implies \eqref{eq:SAbound}.
\end{proof}

Having established the core facts from Section \ref{sec:pathsurplusupper}, and since these are the only facts needed for the analysis in the section, we immediately get a result analogous to Theorem \ref{thm:surplusupper}.
\begin{proposition}
Suppose $m \ge 2$. Let $\mathcal{P}$ be the set of all monotone paths through $[m] \times [m]$ in the modified insertion Surplus Problem. Then,
$$\Pr\left[\max_{P \in \mathcal{P}} \surplus(P) \ge \alpha m \log^{0.75} m \right] \le 2^{-\Omega(\alpha \log m)}.$$
\label{thm:surplusupper-a}
\end{proposition}

Note that Proposition \ref{prop:surplusupper-b} drops the requirement that the first operation is a deletion (rather than an insertion) since this distinction can affect the surplus by at most $O(1)$.

As an immediate corollary, we get: 
\begin{corollary}
Suppose $m \ge 2$. Let $\mathcal{P}$  be the set of all monotone paths through $[m] \times [m]$ in the modified Insertion Surplus Problem. Let $S = \max_{P \in \mathcal{P}} \surplus(P)$. For any positive constant $c_1$, there exists a positive constant $c_2$ such that
    $$\E\left[\max(0, S - c_2 m \log^{0.75} m)\right] \le \frac{1}{m^{c_1}}.$$
\end{corollary}

Finally, since the modified Path Surplus Problem has a maximum path surplus at most $O(t^2 / n)$ larger (in expectation) than the unmodified version of the problem, we also get a bound for the unmodified version of the problem. 

\begin{proposition}
Let $\mathcal{P}$ be the set of all monotone paths through $[0, t] \times [1, R + 1]$ in the (unmodified) Insertion Surplus Problem. Let $m = \sqrt{tR/m}$ and let $S = \max_{P \in \mathcal{P}} \surplus(P)$. If $m \ge 2$, then for any positive constant $c_1$, there exists a positive constant $c_2$ such that
    $$\E\left[\max(0, S - c_2 m \log^{0.75} m)\right] \le \frac{1}{m^{c_1}} + O\left(\frac{t^2}{n}\right).$$
\label{prop:surplusupper-b}
\end{proposition}

The $t^2 / n$ term is negligible so long as $t = n^{o(1)}$. On the other hand, for large $t$ it turns out that we will be okay with using a much weaker inequality. This means that we can directly employ the analysis already in \cite{bender2022linearfull}. Indeed, as an immediate consequence of their Proposition 3, we have:
\begin{lemma}[Insertion Surplus Analysis from \cite{bender2022linearfull}]
Let $\mathcal{P}$ be the set of all monotone paths through $[0, t] \times [1, R + 1]$ in the (unmodified) Insertion Surplus Problem. Let $m = \sqrt{tR/n}$ and let $S = \max_{P \in \mathcal{P}} \surplus(P)$. Let $x \ge 0$ be a parameter, and suppose that $m \ge x \log^c x$ for some sufficiently large positive constant $c$. Then,
$$\E\left[\max(0, S - t/(64x)\right] \le O\left(\frac{1}{m^2}\right).$$
\label{lem:largem}
\end{lemma}
\begin{proof}
    Proposition 3 of \cite{bender2022linearfull} says that, w.h.p.~in $m$, $S \le m \polylog m$. The fact that $m \ge x \log^c x$ (and that $c$ is a sufficiently large positive constant) implies that $m \ge x \log^{c/2} m$.  Combining this with $S \le m \polylog m$ gives  $S \le m^2 (\polylog m) / x \log^{c/2} m$, which for $c$ large enough implies $S = o(m^2 / x)$ (here the $o-$notation is a function of $m$). Observe however, that $o(m^2 / x) = o(\mu / x) \le o(t / x)$, so we have w.h.p.~in $m$ that $S < t / (64 x)$. The $\frac{1}{\poly m}$ probability that $S \ge t / (64 x)$ can contribute at most $1/\poly(m)$ to the expected value of $S$ (since it can contribute at most $1/\poly(m)$ to the expected total number of dots in the grid). Thus the proof is complete.
\end{proof}

Combining this with Proposition \ref{prop:surplusupper-b} gives:
\begin{proposition}
Let $\mathcal{P}$ be the set of all monotone paths through $[0, t] \times [1, R + 1]$ in the (unmodified) Insertion Surplus Problem. Let $x \le n^{o(1)}$ be a parameter and assume that $R \ge \Omega(n/x)$. Let $m = \sqrt{tR/n}$, and let $S = \max_{P \in \mathcal{P}} \surplus(P)$. If $m \ge 2$, then there exists a positive constant $c$ such that
    $$\E\left[\max\left(0, S - c m \log^{0.75} m - \frac{t}{64x}\right)\right] \le O\left(1 / m^2\right).$$
\label{prop:surplusupper-c}
\end{proposition}
\begin{proof}
If $m \ge x \polylog x$, then the result follows from Lemma \ref{lem:largem}. Otherwise, $m \le x \polylog x \le n^{o(1)}$. Since $m = \sqrt{t R /n} \ge \Omega(t / x) \ge t / n^{o(1)}$, it follows that $t = n^{o(1)}$. Thus, in this parameter regime, the result follows from Proposition \ref{prop:surplusupper-b}.
\end{proof}

Later in the paper, it will be helpful to have explicit notation to refer to the quantities studied in Proposition \ref{prop:surplusupper-c}. Given a (not-necessarily-alternating) finite sequence $O = o_1, o_2, \ldots$ of insertions/deletions, and given an interval $[j - t, t] \subseteq [1, n]$, define $\surplus(t_0, O, [j - t, j])$ as follows: Plot the blue dots and red dots for insertions/deletions in $O$ as in the standard insertion Surplus Problem, and let $\mathcal{P}$ be the set of monotonic paths through $[j - t - 1, j] \times [1, |O| + 1]$, starting at $(j -t -1, 1)$ and ending at $(j, |O| + 1)$; then $\surplus(O, [j - t, j])$ is the maximum surplus of any path in $\mathcal{P}$. In this language, we can rewrite Proposition \ref{prop:surplusupper-c} as:

\begin{corollary}
Let $O$ be an alternating sequence of $2R$ insertions/deletions, where $R \le n$, and let $[j - t, j] \subseteq [1, n]$. Let $x \le n^{o(1)}$ be a parameter and let $m = \sqrt{tR/n}$. If $m \ge 2$, then there exists a positive constant $c$ such that
    $$\E\left[\max\left(0, \surplus(O, [j - t, t]) - c m \log^{0.75} m - \frac{t}{64x}\right)\right] \le O\left(1 / m^2\right).$$
\label{cor:surplusupper-c}
\end{corollary}

\subsection{Lower Bound on Insertion Surplus}\label{sec:insertionsurpluslower}

We now turn our attention to proving a lower bound, that is, that there exist operation sequences $R$ for which the expected maximum path surplus is $\Omega(\sqrt{\mu} \log^{0.75} \mu - t^2 / n)$. To simplify discussion, we assume throughout the subsection that the first operation $o_1$ is an insertion, and the final operation $o_{2R}$ is a deletion (this is the opposite of the assumption in the previous subsection). So odd-indexed operations are insertions and even-indexed operations are deletions. Again, this assumption will not affect the final results of the section.

More importantly, we will consider only sequences $R$ with the property that each item is inserted/deleted by \emph{at most one operation}.

As before, to make the Insertion Surplus Problem more closely resemble the Path Surplus problem from Section \ref{sec:pathsurplus}, we will make a sequence of three modifications to the problem that have provably negligible (or decreasing) effect on the maximum path surplus. In fact, in this subsection, we will be able to formally reduce to the exact Path Surplus Problem.

\paragraph{Modification 1: Poissonization.} We Poissonize exactly as in the previous subsection. Once again, the Poissonization affects the maximum path surplus by at most $O(t^2/n)$ in expectation.

\paragraph{Modification 2: Smoothing the dot placement.} We also smooth the dot placement using \emph{almost} the same process as in the previous subsection. Recall, however, that in the previous subsection we were okay with increasing path surplus, but now we are only okay with decreasing path surplus. Thus, we use the following variation of the smoothing process.

As before, we modify the hash functions $h_1, h_2, \ldots$ to output random real numbers in $(0, n]$ (so if previously it produced some integer $\ell$ now it produces a random real number in $(\ell - 1, \ell]$). Moreover, rather than placing the $j$-th blue/red dot for $o_i$ directly at $(h_j(x_i), i)$, we will place it as follows:
\begin{itemize}
    \item If $j$ is odd (i.e., $o_j$ is an insertion), place a blue dot at $(h_j(x_i), b)$ where $b$ is uniformly random in $(j, j + 2]$;
    \item If $j$ is even (i.e., $o_j$ is a deletion), place a red dot at  $(h_j(x_i), b)$ where $b$ is uniformly random in $(j - 1, j + 1]$;
\end{itemize}

We now have:
\begin{lemma}
    Smoothing the dot placement does not increase the maximum path surplus.
\end{lemma}
\begin{proof}
    This is the same proof as for Lemma \ref{lem:smooth1} except that the roles of blue/red dots have reversed. For completeness, give the adjusted proof here.
    
    Prior to smoothing the dot placement, we had without loss of generality that every path traveled along integer grid lines. Let $P$ be such a path.
    
    To analyze the effect of smoothing the dot placement on $P$, let us break the smoothing process into two phases. First, suppose that we perform the following partial smoothing process (the difference is in how we pick $b$):
    \begin{itemize}
    \item If $j$ is odd (i.e., $o_j$ is an insertion), place a blue dot at $(h_j(x_i), b)$ where $b$ is uniformly random in $[j, j + 1)$;
    \item If $j$ is even (i.e., $o_j$ is a deletion), place a red dot at   $(h_j(x_i), b)$ where $b$ is uniformly random in $[j, j + 1)$;
    \end{itemize}
    In this version of smoothing each dot moves from an integer point $(x, y)$ to a real-valued point in $(x - 1, x] \times [y, y + 1)$. This does not change $\surplus(P)$ at all, since the set of points beneath $P$ is precisely the same as before.

    Now, to get from the partial smoothing process to the full smoothing process, we can perform the following additional modification:
    \begin{itemize}
        \item Each blue dot that is currently in some position $(x, y)$ gets moved to $(x, y + 1)$ with probability $1/2$.
        \item Each red dot that is currently in some position $(x, y)$ gets moved to $(x, y - 1)$ with probability $1/2$.
    \end{itemize}
    This can only decrease $\surplus(P)$ since we are moving blue dots up and red dots down. On the other hand, when we combine the partial smoothing process with this random movement step, we get the full smoothing process, which completes the proof of the lemma.
\end{proof}

As in the previous subsection, smoothing ensures that the number of blue (or red) dots in a given cell $(a - 1] \times [b, b+ 1)$ is distributed as $\Pois(1/(2n))$.
\begin{lemma}
    Let $a \in \{1, 2, \ldots, t\}$ and $b \in \{1, 2, \ldots, R\}$. Consider the geometric region $G = (a - 1, a] \times [b, b + 1)$. Then the number of blue dots $B$  and the number of red dots $R$ in $G$ are both Poisson random variables with mean $1/(2n)$, and the dots in $G$ are placed at uniformly random (and mutually independent) positions in $G$.
    \label{lem:Poissoncells2}
\end{lemma}
\begin{proof}
    Suppose $b$ is even. Prior to smoothing, there were $\Pois(1/n)$ red points (resp.~blue points) at $(a, b)$ (resp.~$(a, b - 1)$), each of which now has a 50\% probability of being placed uniformly at random in $G$ (due to smoothing). 

    Suppose $b$ is odd. Prior to smoothing, there were $\Pois(1/n)$ red points (resp.~blue points) at $(a, b + 1)$ (resp.~$(a, b)$), each of which now has a 50\% probability of being placed uniformly at random in $G$ (due to smoothing). 
\end{proof}

In fact, because each operation $o_i$ is on a different key $x_i$, we can make a stronger claim:
\begin{lemma}
The blue points (resp.~red points) are generated by a Poisson random process with density $1/(2n)$ in the grid $[0, t] \times [1, R + 1]$. 
\label{lem:fullypoisson}
\end{lemma}
\begin{proof}
    Define $X_{a, b}$ to be the set of blue/red dots that land in \defn{cell} $(a - 1, a] \times [b, b + 1)$. By the argument in Lemma \ref{lem:Poissoncells2}, the arrivals within each cell are Poisson. Therefore, it suffices to show that the $X_{a, b}$s are mutually independent across $a$ and $b$.

    As in the previous subsection, Poissonization gives us vertical independence, guaranteeing independence across $a$. It therefore suffices to show that, for any given $a$, the $X_{a, b}$s are mutually independent across $b$.

    This is where we make use of the fact that each operation $o_i$ is on a different key $x_i$. Each deletion $o_i$ (for even $i$) generates the blue-dot Poisson arrivals for $(a - 1, a] \times [i - 1, i + 1)$ and each insertion $o_i$ (for odd $i$) generates the red-dot Poisson arrivals for $(a - 1, a] \times [i, i + 2)$. Because the operations are on different keys, the arrival processes are independent. This ensures mutual independence for the $X_{a, b}$s across $b$, as desired.
\end{proof}

\paragraph{Modification 3: Rescaling the grid.} Finally, as in the previous subsection, we can rescale the $[0, t] \times [1, R + 1]$ grid to be on $[m] \times [m]$ where $m = \sqrt{\mu}$. This has no effect on the path surplus, but changes the density of dots so that Lemma \ref{lem:fullypoisson} becomes:

\begin{lemma}
The blue points (resp.~red points) are generated by a Poisson random process with density $1$ in the grid $[0, t] \times [1, R + 1]$. 
\label{lem:fullypoisson2}
\end{lemma}

\paragraph{Analysis.} Lemma \ref{lem:fullypoisson2} tells us that the modified problem is actually \emph{exactly} the Path Surplus Problem studied in Section \ref{sec:pathsurplus}. We therefore have that:

\begin{lemma}
   The expected maximum path surplus in the modified problem is $\Omega(m \log^{0.75} m)$. 
\end{lemma}

Since the transformations increase the expected surplus by $O(t^2/n)$, it follows that:

\begin{proposition}
Suppose that every operation is on a distinct key and that $t = n^{o(1)}$. Let $\mathcal{P}$ be the set of all monotone paths through $[0, t] \times [1, R + 1]$ in the (unmodified) insertion Surplus Problem. Then, setting $m = \sqrt{tR/n}$, we have
$$\E\left[\max_{P \in \mathcal{P}} \surplus(P)\right] \ge \Omega\left(m \log^{0.75} m\right).$$
\label{prop:surpluslower-b}
\end{proposition}
Note that, as in the previous section, our final result drops the requirement about whether the first operation is an insertion/deletion since once again this distinction can affect the maximum path surplus by at most $O(1)$.

Rewriting Proposition \ref{prop:surpluslower-b} in terms of the $\surplus(O, [j - t, t])$ notation (introduced at the end of Subsection \ref{sec:insertionsurplusupper}), we get the following corollary:
\begin{corollary}
Let $O$ be an alternating sequence of $2R$ insertions/deletions, where each operation is on a distinct key, and suppose that $t \le n^{o(1)}$. Let $[j - t, t] \subseteq [1, n]$, and let $m = \sqrt{tR/n}$. Then 
$$\E\left[\surplus(O, [j -t, t])\right] \ge \Omega\left(m\log^{0.75} m\right).$$
\label{cor:surpluslower-c}
\end{corollary}

\section{Analysis of Linear Probing}\label{sec:linprobing}

We are now prepared to derive tight bounds on the amortized expected cost of ordered linear probing with tombstones. Consider an ordered linear probing hash table on $n$ slots that implements deletions with tombstones, and performs rebuilds every $R$ insertions. Consider two consecutive rebuilds taking place at times $ t_0 $ and $ t_1 $, and suppose that the load factor never exceeds $1-1/x$ during $[t_0, t_1]$ (recall that the load factor does not count tombstones towards the load). There are $R$ insertions that take place during $[t_0, t_1]$.

Define the \defn{crossing numbers} $ c_1, c_2,\ldots, c_n $ so that $ c_i $ is the number of times that an insertion with a hash smaller than $ i $ either (a) uses a tombstone left by a key that had hash at least $ i $; or (b) uses a free slot in a position greater than or equal to $ i $.

Bender et al.~\cite{bender2022linearfull} establish a tight relationship between insertion/query time and the average crossing number. Roughly speaking, this comes from the fact that, if an insertion of an item $u$ uses a tombstone/free-slot with hash/position $h(u) + k$, then the insertion is guaranteed to have incremented $k$ crossing numbers. This means that, rather than directly analyzing insertion times, one can instead analyze crossing numbers, and then extract the insertion (and query) times from that.

\begin{lemma}[Upper Bound in Terms of Crossing Numbers]
Let $i \in [n]$ and $u \in I$ be uniformly random. Then, the expected amount of time spent on insertion $u$ is 
$$O\left(x + (1 + n/R) \E[c_i]\right).$$
Moreover, the expected amount of time spent on any query or deletion is 
$$O\left(x + \E[c_i]\right).$$
\label{lem:timetocrossingnumber}
\end{lemma}
\begin{proof}
See Lemmas 13 and 14 of \cite{bender2022linearfull}.
\end{proof}

\begin{lemma}[Lower Bound in Terms of Crossing Numbers]
Suppose $R \le n$. Let $i \in [n]$ and $u \in I$ be uniformly random. Then, the expected amount of time spent on insertion $u$ is 
\begin{equation}\Omega\left(\frac{n}{R} \E[c_i].\right)
\label{eq:crossa}
\end{equation}
Moreover, the expected amount of time spent on any negative query performed at time $t_1$ is
\begin{equation} \Omega\left(\E[c_i]\right).
\label{eq:crossb}
\end{equation}
\label{lem:timetocrossingnumber2}
\end{lemma}
\begin{proof}
Both Equations \eqref{eq:crossa} and \eqref{eq:crossb} are established in \cite{bender2022linearfull} in the proof of Theorem 1 (using their Lemma 14).
\end{proof}

Thus, to analyze linear probing, the key technical result that we must prove is that:
\begin{proposition}
Let $x$ and $n$ be parameters satisfying $x \le n^{o(1)}$, and let $R \le n$ be the rebuild-window size. Suppose $\beta = n/R$ satisfies $\beta \le x$. Supposing that the load factor never exceeds $1 - 1/x$ during time interval $[t_0, t_1]$, we have that
$$\E\left[\sum_s c_s\right] \le O(Rx \log^{1.5} x + nx).$$
Moreover, there exists a workload for which this is tight. 
\label{prop:crossing}
\end{proposition}

Assuming Proposition \ref{prop:crossing}, we can complete a tight analysis of linear probing. Namely, we can prove Theorem \ref{thm:main}, restated here:
\mainthm*
\begin{proof}
It suffices to analyze the operations between two consecutive rebuilds. Note that the amortized expected cost of the rebuilds themselves is $O(n/R) = O(\beta)$ time per insertion/deletion, which is negligible compared to the other terms that we will be considering.

By Lemma \ref{lem:timetocrossingnumber} and Proposition \ref{prop:crossing}, the amortized expected insertion time is 
$$O\left(x + \frac{n}{R} \cdot \frac{Rx \log^{1.5} x + nx}{n}\right) = O\left(x \log^{1.5} x  + \beta x\right).$$
By Lemma \ref{lem:timetocrossingnumber2} and Proposition \ref{prop:crossing}, there exists a workload for which the amortized expected insertion time is
$$\Omega\left( \frac{n}{R} \cdot  \frac{Rx \log^{1.5} x + nx}{n}\right) = \Omega\left(x \log^{1.5}x  + \beta x\right).$$
By Lemma \ref{lem:timetocrossingnumber} and Proposition \ref{prop:crossing}, the expected query time is 
$$O\left(x + \frac{Rx \log^{1.5} x + nx}{n}\right) = O\left(x + \frac{x \log^{1.5} x}{\beta}\right).$$
By Lemma \ref{lem:timetocrossingnumber2} and Proposition \ref{prop:crossing}, there exists a workload and a query for which the expected query time is 
$$\Omega\left(\frac{Rx \log^{1.5} x + nx}{n}\right) = \Omega\left(x + \frac{x \log^{1.5} x}{\beta}\right).$$
\end{proof}

\begin{corollary}
The optimal rebuild-window size is $R = \Theta(n / \log^{1.5} x)$, resulting in wort-case amortized expected insertion cost of $\Theta(x \log^{1.5} x)$ and an expected query time of $\Theta(x)$.
\end{corollary}

\subsection{Proof of Proposition \ref{prop:crossing}}\label{sec:finalproposition}

Define the \defn{surplus} $\surplus(i, j)$ of an interval $[i, j] \subseteq [1, n]$ to be (using the notation established at the end of Subsection \ref{sec:insertionsurplusupper}) $\surplus(O, [i, j])$ where $O$ is the sequence of insertions/deletions performed during the rebuild window $[t_0, t_1]$. And define the \defn{free-slot contribution} $\free(i, j)$ to be the number of free slots in $[i, j]$ at time $t_0$. Throughout this section, we shall assume for context that $n, \beta, R$ are as defined in Proposition \ref{prop:crossing}.

The following lemma tells us how to interpret crossing numbers in terms of surpluses and free-slot contributions. At a high level, what is says is that the way a large crossing number $c_s$ can occur is that there is some interval $[i, s - 1]$ such that $\surplus(i, s - 1)$ overpowers $\free(i, s - 1)$ by $c_s$.
\begin{lemma}[From Crossing Numbers to Surplus]
For any given $s \in [n]$, we have that 
$$c_s = \max_{i < s} \left( \surplus(i, s - 1) - \free(i, s - 1)\right).$$
\label{lem:crossingnumbertosurplus}
\end{lemma}
\begin{proof}
This follows from Lemmas 9 and 10 of \cite{bender2022linearfull}.
\end{proof}

\begin{lemma}[Lower-Bound Side of Proposition \ref{prop:crossing}]
There exists a workload such that
$$\E\left[\sum_s c_s\right] \ge \Omega(Rx \log^{1.5} x + nx).$$
\label{lem:crossinga}
\end{lemma}
\begin{proof}
We begin by designing a workload with 
\begin{equation}\E\left[\sum_s c_s\right] \ge \Omega(nx).
\label{eq:lbtriv}
\end{equation}
In order for a workload to guarantee this, it turns out that all we need is for the workload to start at a load factor of $1 - 2/x$, and to begin with $n/x \le R$ insertions in a row. Each of these insertions takes expected time $\Theta(x^2)$ and contributes $\Theta(x^2)$ in expectation to $\sum_s c_s$ (by the classical tombstone-free analysis of linear probing). Summing over the $n/x$ insertions gives \eqref{eq:lbtriv}.

The more interesting challenge is to achieve a lower bound of the form
$$\E\left[\sum_s c_s\right] \ge \Omega(Rx \log^{1.5} x).$$
Note that we only need to show this for cases where $Rx \log^{1.5} x \ge nx$, so we can assume that $\beta \le O(\log^{1.5}x) \le x^{1 - \Omega(1)}$ (recall that $\beta$ is defined as $n / R$).
To prove the lower bound, we use the following workload: start at a load factor of $1 - 1/x$, and then perform an alternating sequence of $2R$ insertions/deletions, where each operation is on a different key. Critically, this is the setting where we can apply Corollary \ref{cor:surpluslower-c} to analyze insertion surplus.

Set $t = \frac{x^2}{q \beta} \log^{1.5} x$ for some parameter $q \le x^{o(1)}$ (that we will later set to a large constant). For each $s \in [n]$, we have that
\begin{align*}
\E[c_s] & = \E\left[\max_{i < s} \left( \surplus(i, s - 1) - \free(i, s - 1)\right)\right] \\
& \ge  \E\left[\surplus(s - t, s - 1) - \free(s - t, s - 1)\right] \tag{by Lemma \ref{lem:crossingnumbertosurplus}}\\
& =  \E\left[\surplus(s - t, s - 1)\right] - t/x \tag{since load factor is $1-1/x$, and by symmetry across $s$} \\
& \ge  \Omega\left(\sqrt{t/\beta} \log^{0.75} \sqrt{t/\beta}\right) - t/x \tag{by Corollary \ref{cor:surpluslower-c}} \\
& \ge  \Omega\left(\sqrt{t/\beta} \log^{0.75} x\right) - t/x \tag{since $\beta \le x^{1 - \Omega(1)}$ and $t \ge \Omega(x)$} \\
& \ge  \Omega\left(\frac{x}{\sqrt{q} \beta} \log^{1.5} x - \frac{x}{q \beta} \log^{1.5} x\right) \tag{by expanding $t$}. \\
\end{align*}
Setting $q$ to be a sufficiently large positive constant gives 
$$\E[c_s] \ge \Omega\left(\frac{x}{\beta} \log^{1.5} x\right),$$
implying that
$$\E\left[\sum_s c_s\right] \ge \Omega\left(\frac{n x}{\beta} \log^{1.5} x\right) = \Omega(Rx \log^{1.5} x),$$
as desired.
\end{proof}


To prove the upper-bound side of Proposition \ref{prop:crossing}, we will first need a few more lemmas.

\begin{lemma}[Free-Slot Lower Bound]
For $i \ge 0$, let $ F_i $ be the number of free slots in positions $ [n - i, n] $ at time $t_0$. Then,
$$\E\left[\max_{i \ge 0} \left(\frac{i}{4x} - F_i\right)\right] \le O(x).$$
\label{lem:freeslots}
\end{lemma}
\begin{proof}
We can reduce to power-of-two values of $ i $ with the observation that
$$\max_{i \ge 0} \left(\frac{i}{4x} - F_i\right) \le \max_{i = 2^t}  \left(\frac{i}{2x} - F_i\right).$$
We can further reduce to $i \ge x^2$ with the observation that
$$\max_{i = 2^t}  \left(\frac{i}{2x} - F_i\right) \le \max_{i = 2^t \ge x^2}  \left(\frac{i}{2x} - F_i\right) + O(x).$$
This, in turn, is at most
$$O(x) + \sum_{i = 2^t \ge x^2}  \max\left(0, \frac{i}{2x} - F_i\right).$$
Therefore to complete the proof, it suffices to show that for each individual $i = 2^t \ge x^2$, we have
\begin{equation}
\E\left[\max\left(0, \frac{i}{2x} - F_i\right)\right] \le \frac{i}{x} \cdot 2^{-\Omega(\sqrt{i} / x)}.
\label{eq:freeslotsbyi}
\end{equation}

Let $A_i$ be the number of elements $u$ in the hash table at time $ t_0 $ with hashes $h(u) \in [n - i, n]$, and let $B_i$ be the number of elements $u$ in slots $[n - i, n]$ of the hash table at time $t_0$ with hashes $h(u) < n - i$. Then 
$$F_i \ge i - A_i - B_i.$$
Since $A_i$ is a sum of independent indicator random variables with mean $\Theta(i)$, we have by a Chernoff bound that
$$\Pr[|A_i - \E[A_i]| \ge k \sqrt{i}] \le 2^{-\Omega(k)}.$$
We can also bound $B_i$, using Corollary 2 of \cite{bender2022linearfull}, to get
$$\Pr[B_i  \ge k x] \le 2^{-\Omega(k)}.$$
It follows that
$$\Pr[F_i \le \E[i - A_i - B_i] - k(\sqrt{i} + x)] \le 2^{-\Omega(k)}.$$
Since $\E[i - A_i - B_i] \ge i - (1 - 1/x)i - O(x) \ge i/x - O(x),$ it follows that
$$\Pr[F_i \le i/x - k(\sqrt{i} + x)] \le 2^{-\Omega(k)}.$$
Recalling that $i \ge x^2$, we have that
$$\Pr[F_i \le i/x - k\sqrt{i}] \le 2^{-\Omega(k)},$$
and thus that
$$\Pr[F_i \le i/(2x)] \le 2^{-\Omega(\sqrt{i} / x)}.$$
This means that 
$$\E\left[\max\left(0, \frac{i}{2x} - F_i\right)\right] \le \frac{i}{x} \cdot 2^{-\Omega(\sqrt{i} / x)},$$
as desired.

\end{proof}

Call a sequence of insertions/deletions \defn{hovering} if the sequence alternates between insertions and deletions, and starts with a hash table out load factor $1 - 1/x$. It turns out that, to prove the upper-bound side of Proposition \ref{prop:crossing}, it suffices to focus on hovering workloads exclusively.

\begin{lemma}[Hovering Workloads are WLOG]
Let $A$ be the set of elements present at time $t_0$, and let $ O $ be the sequence of insertions/deletions that take place in $[t_0, t_1]$ (never exceeding a load factor of $1 - 1/x$). 

Then there exists a set $A'$ of size at most $(1 - 1/x)n$, and a hovering sequence of insertions/deletions $S'$ such that if $A'$ and $O'$ were used in place of $A$ and $O$, then the crossing numbers $c_i'$ would satisfy 
$$\E\left[\sum_i c_i\right] \le \E\left[\sum_i c'_i\right] + O(nx).$$ 
\label{lem:wlogworkload}
\end{lemma}
\begin{proof}
This is established in the proof of Proposition 7 of \cite{bender2022linearfull}.
\end{proof}

The next two lemmas are consequences of \ref{cor:surplusupper-c}.
\begin{lemma}[Consequence 1 of Corollary \ref{cor:surplusupper-c}]
Suppose $O$ is an alternating sequence of $2R$ insertions/deletions, where $n/x \le R \le n$. Let $j \in [n]$ and let $x \le n^{o(1)}$ be a parameter. Set $m_t = \sqrt{tR/n}$ for $t \ge 1$. Then there exists a positive constant $c$ such that
    $$\E\left[\max_{t \ge 0} \left(\surplus(j - t, j) - c m_t \log^{0.75} m_t - \frac{t}{32x}\right)\right] \le O\left(x\right).$$
\label{lem:surplusupper-d}
\end{lemma}
\begin{proof}

We begin by observing that we can reduce to power-of-two values of $t$, since
\begin{align*}
& \E\left[\max_{t \ge 0} \left(\surplus(j - t, j) - c m_t \log^{0.75} m_t - \frac{t}{32x}\right)\right] \\
& \le \E\left[\max_{t = 2^i} \left(\surplus(j - t, j) - c m_{t/2} \log^{0.75} m_{t/2} - \frac{t}{64x}\right)\right] \tag{since $\surplus(j - t, j)$ is non-decreasing in $t$}\\
& \le \E\left[\max_{t = 2^i} \left(\surplus(j - t, j) - \frac{c}{2} m_{t} \log^{0.75} m_{t} - \frac{t}{64x}\right)\right]. \\
\end{align*}

Set $\overline{t}$ to be the smallest $t$ such that $\frac{c}{2} m_{t} \log^{0.75} m_{t} + \frac{t}{64x} \ge x$. Note that this inequality is bottlenecked by the term involving $m_t$ since
$$\frac{c}{2} m_{t} \log^{0.75} m_{t} + \frac{t}{64x} = \Theta(\sqrt{tR/n} \log^{0.75} (tR/n) + t/x),$$
and since $n/R \le x$. Therefore $m_{\overline{t}} = \Theta(x / \log^{0.75} x)$. Critically, this means without loss of generality that $m_{\overline{t}} \ge 2$ (since WLOG $x$ is at least a large positive constant), which will allow us to apply Corollary \ref{cor:surplusupper-c} later in the proof.

Define $f(t) = \surplus(j - t, j) - \frac{c}{2} m_{t} \log^{0.75} m_{t} - \frac{t}{64x}$. By the definition of $\overline{t}$, we have for $t < \overline{t}$ that $f(t) \le f(\overline{t}) + \frac{c}{2} m_{\overline{t}} \log^{0.75} m_{\overline{t}} - \frac{\overline{t}}{64x} \le f(\overline{t}) + O(x)$. Thus, we have that
\begin{align*}
    & \E\left[\max_{t = 2^i} \left(\surplus(j - t, j) - \frac{c}{2} m_{t} \log^{0.75} m_{t} - \frac{t}{64x}\right)\right] \\
    & = \E\left[\max_{t = 2^i} f(t)\right]\\
    & \le \E\left[\max_{t = 2^i \ge \overline{t}} f(t)\right] + O(x).\\
    & \le \sum_{t = 2^i \ge \overline{t}} \E\left[f(t)\right] + O(x).
\end{align*}
By Corollary \ref{cor:surplusupper-c} (which we can apply since we know that $m_t \ge m_{\overline{t}} \ge 2$), and assuming $c$ is a sufficiently large positive constant, this sum is at most
\begin{align*}
& \sum_{t = 2^i \ge \overline{t}} O(1/m_t^2) + O(x). \\
& = \sum_{i \ge 0} O\left(\frac{1}{\sqrt{2^i \overline{t} R/n}^2}\right) + O(x). \\
& = \sum_{i \ge 0} O\left(\frac{1}{2^i m_{\overline{t}}^2}\right) + O(x). \\
& = O(1/ m_{\overline{t}}^2) + O(x)\\
& = O(x),
\end{align*}
as desired.
\end{proof}

\begin{lemma}[Consequence 2 of Corollary \ref{cor:surplusupper-c}]
    Suppose $O$ is an alternating sequence of $2R$ insertions/deletions, where $n/x \le R \le n$. Let $j \in [n]$ and let $x \le n^{o(1)}$ be a parameter. Then there exists a positive constant $c$ such that
    $$\E\left[\max_{t \ge 0} \left(\surplus(j - t, j) - \frac{t}{16x}\right)\right] \le O\left(\frac{x}{\beta}\log^{1.5} x + x\right).$$
    \label{lem:surplusupper-e}
\end{lemma}
\begin{proof}
Define $m_t = \sqrt{tR/n}$, and let us take a moment to find the positive solution to the equation
\begin{equation}m_t \log^{0.75} m_t = \Theta\left(\frac{t}{x}\right).\label{eq:mtequi}\end{equation}
Since $\sqrt{tR/n} = m_t$, we have $t = m_t^2 n / R$, so the equation reduces to
$$m_t \log^{0.75} m_t = \Theta\left(\frac{m_t^2 n}{R x}\right),$$
which implies
$$\frac{x}{\beta} = \Theta\left(\frac{m_t}{\log^{0.75} m_t}\right).$$
It follows that, when \eqref{eq:mtequi} holds, we have
$$m_t \log^{0.75} m_t = \Theta\left(\frac{x}{\beta} \log^{1.5} \frac{x}{\beta}\right) \le O\left(\frac{x}{\beta}\log^{1.5} x\right).$$
The reason we care about this is that it implies for any positive constants $c_1$ and $c_2$, and for any $m_t$, that
$$c_1 m_t \log^{0.75} m_t \le O\left(\frac{x}{\beta}\log^{1.5} x\right) + \frac{t}{c_2x}.$$

Thus, for any positive constant $c$, we have
\begin{align*}
    & \E\left[\max_{t \ge 0} \left(\surplus(j - t, j) - \frac{t}{32x}\right)\right] \\
    & \le O\left(\frac{x}{\beta}\log^{1.5} x\right) + \E\left[\max_{t \ge 0} \left(\surplus(j - t, j) - c m_t \log^{0.75} m_t - \frac{t}{16x}\right)\right].
\end{align*}
The lemma then follows by Lemma \ref{lem:surplusupper-d}
\end{proof}

We can now complete the proof of the upper-bound side of Proposition \ref{prop:crossing}
\begin{lemma}[Upper-Bound Side of Proposition \ref{prop:crossing}]
Let $\beta = n/R$ satisfy $\beta \le x$. Supposing that the load factor never exceeds $1 - 1/x$ during time interval $[t_0, t_1]$, we have that
$$\E\left[\sum_s c_s\right] \le O(Rx \log^{1.5} x + nx).$$
\label{lem:crossingupper}
\end{lemma}
\begin{proof}
    By Lemma \ref{lem:wlogworkload}, we can assume without loss of generality that the operations during $[t_0, t_1]$ are a hovering workload. 
    
    Let us focus on bounding some fixed $s$. Let $O$ denote the insertions/deletion sequence in $[t_0, t_1]$ (which, again, we can assume is a hovering workload), and let $F_i = \free(s - 1 - i, s - 1)$. By Lemma \ref{lem:crossingnumbertosurplus}, we have that
    \begin{align*}
        \E[c_s] & = \E\left[\max_{i \ge 0} \left( \surplus(s - i, s - 1) - F_i\right)\right] \\
        & \le \E\left[\max_{i \ge 0} \left(\surplus(s - i, s - 1) - \frac{i}{4x}\right) + \max_{i \ge 0} \left(\frac{i}{4x} - F_i\right)\right] \\
        & \le \E\left[\max_{i \ge 0} \left(\surplus(s - i, s - 1) - \frac{i}{4x}\right)\right] + \E\left[\max_{i \ge 0} \left(\frac{i}{4x} - F_i\right)\right] \\
        & \le \E\left[\max_{i \ge 0} \left(\surplus(s - i, s - 1) - \frac{i}{4x}\right)\right] + O(x) \tag{by Lemma \ref{lem:freeslots}} \\
        & \le  O\left(\frac{x}{\beta}\log^{1.5} x + x\right). \tag{by Lemma \ref{lem:surplusupper-e}}
    \end{align*}
    Finally, summing over $s \in \{1, \ldots, n\}$ gives the desired bound on $\E[\sum_s c_s]$.
\end{proof}

\section{Extending to Unordered Linear Probing}\label{sec:unordered}

In this section, we show how to extend Theorem \ref{thm:main} to analyze a `common-case' workload for \emph{unordered} linear probing. What distinguishes this workload from a worst-case workload will be that: each query and each deletion is to a \emph{random} element out of those present; each insertion is to a new element, never before inserted; and the hash table hovers at $1 - 1/x$ full.

Theorem \ref{thm:unordered} says that, under these conditions, the time bounds for unordered linear probing are the same as those for ordered linear probing. Thus, one can view the distinction of whether the hash table is \emph{ordered} as being about whether one wishes for worst-case guarantees (using ordered linear probing) or average-case guarantees (using unordered linear probing).


\thmunordered*

It is worth noting that, in Theorem \ref{thm:unordered}, both the restriction that queries/deletions are to random elements (out of those present) and the restriction that insertions are to new elements are fundamentally necessary for the theorem to hold. If either of these restrictions are relaxed, then one can force $\Theta(x^2)$-time operations by either repeatedly inserting/deleting the same element over and over, or repeatedly querying the first element to be inserted during the rebuild window. 

As an immediate corollary of Theorem \ref{thm:unordered}, we get:
\corunordered*

In the rest of the section, we will prove Theorem \ref{thm:unordered}.
At a high level, the analysis will proceed in the same way as it did for analyzing ordered linear probing, although, as we shall see, there will be several significant new complications when it comes to bounding insertion surpluses.

Let us focus on the operations that take place between two rebuilds. We will use $I$ to refer to the set of elements present at the start of the time window, and $\overline{O}$ to refer to the sequence of $2R$ alternating insertions/deletions that occur.

Define the \defn{unordered crossing numbers} $ \overline{c}_1, \overline{c}_2,\ldots, \overline{c}_n $ so that $ \overline{c}_i $ is the number of times that an insertion with a hash smaller than $ i $ uses a tombstone or free slot in position at least $i$. One convenient feature of unordered linear probing (as opposed to ordered linear probing) is that there is a direct relationship between the unordered crossing numbers and the total time spent on insertions, namely that the total time $T_{\text{ins}}$ spent on insertions in $\overline{O}$ is exactly
\begin{equation}
T_{\text{ins}} = \sum_i c_i.
\label{eq:Tins}
\end{equation}

To analyze the unordered crossing numbers, it will be helpful to define what we call the \defn{unordered insertion-surplus} $\overline{\surplus}(I, O, [j - t, j))$ of an interval $[i, j]$. This is calculated by placing blue/red dots as follows:
\begin{itemize}
\item a blue dot at $(h(x), i)$ if the $i$-th insertion in $\overline{O}$ inserts an element $x$;
\item and a red dot at $(k, i)$ if the $i$-th deletion in $\overline{O}$ removes an element that was in \emph{position} $k$. 
\end{itemize}
Now consider the $t \times R$ sub-grid containing dots $(a, b)$ satisfying $a \in [0, t]$ and $b \in [j - t - 1, j)$. The unordered insertion surplus $\overline{\surplus}(I, O, [j - t, j))$ is defined to be the maximum path surplus (as defined in Section \ref{sec:pathsurplus}) or any monotonic path through the grid. 

Conveniently, unordered insertion surpluses have precisely the same relationship to unordered crossing numbers as do the ordered versions of the same quantities. 

\begin{lemma}[From Unordered Crossing Numbers to Unordered Surplus]
Consider some $j \in [n]$, and for each $t \in [n]$ define $F_0(t)$ to be the number of free slots in the interval $[j - i, j)$ at the beginning of the rebuild window.\footnote{Note that $F_0(t)$ is independent of how elements are ordered within each run \cite{Peterson57}.} Then, we have that
$$c_s = \max_{i < j} \left( \overline{\surplus}(I, O, [j - t, j)) - F_0(t)\right).$$
\label{lem:unorderedcrossingnumbertosurplus}
\end{lemma}
\begin{proof}
This follows from exactly the same arguments as in the proofs of Lemmas 9 and 10 in \cite{bender2022linearfull}. The only difference is that now the proof takes place in position space rather than hash space---that is, whereas in ordered linear probing each insertion of an element $x$ looks through the elements with hashes $h(x), h(x) + 1, \ldots$ until it finds a free slot or tombstone, in unordered linear probing the insertion looks through the positions $h(x), h(x) + 1, \ldots$. Besides this distinction, the proof is exactly the same as in \cite{bender2022linearfull}.
\end{proof}

The main step in completing the analysis will be to prove the following bound on the quantity from \ref{lem:unorderedcrossingnumbertosurplus}.

\begin{proposition}
  Let $R = n / \beta$, where $\beta$ is at least a sufficiently large positive constant and is at most $O(x)$. Consider an initial set $I$ of $(1 - 1/x)n$ elements, right after a rebuild, and consider an alternating sequence $\overline{O}$ of $2R = 2 n / \beta$ insertions/deletions, where each insertion uses a never-before-inserted element, and each deletion uses a random element out of those present. Then 
      $$\E\left[\max_t \left(\overline{\surplus}(I, \overline{O}, [j - t, j)) - F_0(t)\right)\right] = \Theta\left(\frac{R}{n} x \log^{1.5} x + x\right).$$
      \label{prop:translate}
\end{proposition}

As we shall see, Proposition \ref{prop:translate} is actually quite tricky to prove. This is, in part, because of the following reason: in the definition of unordered path surplus, the position of each red dots is determined by the \emph{position} in which a deletion takes place rather than the \emph{hash}. This means that the red dots are not independent. In particular, there are positive correlations: regions of the hash table that have more free slots will systematically accrue red dots slower than regions that have fewer free slots. 

We will prove Proposition \ref{prop:translate} in Subsection \ref{sec:translate}. Before we do that, however, let us show how to use the proposition in order to complete the proof of Theorem \ref{thm:unordered}.

\begin{proof}[Proof of Theorem \ref{thm:unordered}]
By \eqref{eq:Tins}, the total time $T_\text{ins}$ spent on the $R$ insertions is
$$T_{\text{ins}} = \sum_{i \in [n]} \overline{c}_i,$$
which, by Lemma \ref{lem:unorderedcrossingnumbertosurplus} and Proposition \ref{prop:translate}, has expectation
\begin{align*} \E[T_{\text{ins}}] & = \Theta\left(\sum_{i = 1}^n  \left(\frac{R}{n} x \log^{1.5} x + x\right)\right) \\
& = \Theta\left( Rx \log^{1.5} x + nx\right).
\end{align*}
The amortized expected time per insertion is therefore 
\begin{align*} \frac{1}{R} \E[T_{\text{ins}}] & = \Theta(x \log^{1.5} x + nx/R) = \Theta(x \log^{1.5} x + \beta^{-1}x).
\end{align*}

To complete the proof, let us consider a query that occurs immediately after the $i$-th insertion during the rebuild window. There are two cases.

The first case is that the query is to an element $u \in I$. This case occurs with probability $\Theta(1)$. And, if we condition on this case,  then the expected time to perform the query will be $O(x)$, since the expected time to perform a random element in any freshly-constructed linear-probing hash table, with load factor $1 - 1/x$, is $O(x)$ \cite{Knuth63}.

The second case is that the query is to an element $u \not\in I$. This case occurs with probability $\Theta(i/n)$. Moreover, if we condition on this case, and if we define $T_{\text{ins}}(i)$ to be the total time spent on the first $i$ insertions, then the expected query time is 
    $$\Theta(\E[T_{\text{ins}}(i)] / i).$$
The total contribution of this second case to the expected query time is therefore
    \begin{equation} \Theta(\E[T_{\text{ins}}(i)] / n).
    \label{eq:qt}
    \end{equation}

Since \eqref{eq:qt} increases monotonically in $i$, the worst case is $i = n / \beta$. Setting $i = n / \beta$ and combining Cases 1 and 2, we can conclude that the worst-case expected time of querying a random element is 
$$\Theta(x + \E[T_{\text{ins}}] / n),$$
where the first term comes from Case 1 and the second term comes from Case 2.

Applying our bound on $\E[T_{\text{ins}}] = \Theta(R x \log^{1.5} x + nx)$ gives a query time of 
$$\Theta(x + \beta^{-1} x \log^{1.5} x),$$
as desired.
\end{proof}

\subsection{Proof of Proposition \ref{prop:translate}}\label{sec:translate}

In this subsection, we prove Proposition \ref{prop:translate}. To setup the proof, it is helpful to think of $\overline{O}$ as being generated as follows. The $i$-th insertion inserts an element $x_i$, but then the $i$-th deletion picks who to delete with the following process:
\begin{itemize}
\item Sample a random position $k_i$ in the hash table.
\item If there is an element $y$ in position $k$, then set the $i$-th deletion $y_i = y$.
\item Otherwise, select a random element $y$ out of those present and set $y_i = y$. In this case, position $k_i$ is referred to as a \defn{failed deletion target} and the position containing $y$ is called a \defn{second-choice deletion target}.
\end{itemize}
Notice that this construction, despite being slightly round-about, still selects each $y_i$ uniformly at random from the elements present. Thus it is a valid construction for $\overline{O}$.

Now construct an alternative operation sequence $O$ in which the $i$-th insertion still inserts $x_i$, but the $i$-th deletion deletes some $z_i$ with $h(z_i) = k_i$. This may result in us deleting $z_i$s that are not present in either the initial set of elements $I$ or in $\{x_i\}_i$, but that's alright---we will not actually be interested in thinking about $O$ as a sequence of operations, we will just be interested in the quantity $\surplus(O, [j - t, j))$, as defined in Section \ref{sec:insertionsurplus}. 

A critical feature of $O$ is that, since the $k_i$s were drawn randomly, each deletion $z_i$ has a random hash in $h(z_i) = k_i \in [n]$. Thus we will be able to obtain upper and lower bounds on $\surplus(O, [j - t, j))$ via 
Corollaries \ref{cor:surpluslower-c} and \ref{cor:surpluslower-c}.

Our main task will be to develop a formal relationship between $\overline{\surplus}(I, \overline{O}, [j - t, j))$ and $\surplus(O, [j - t, j))$. We begin with the following observation:
\begin{lemma}
  Consider some interval $[j - t, j)$. Define $A$ to be the number of failed deletions in $\overline{O}$ with failed targets in $[j - t, j)$, and define $B$ to be the number of failed deletions in $\overline{O}$ with replacement targets in $[j -t, j)$. Then,
  $$\surplus(O, [j - t, j)) - B \le \overline{\surplus}(I, \overline{O}, [j - t, j)) \le \surplus(O, [j - t, j)) + A.$$
    \label{lem:phantom}
\end{lemma}
\begin{proof}
   The blue/red dots that determine $S_1 = \overline{\surplus}(I, \overline{O}, [j - t, j)) $ and $S_2 = \surplus(O, [j - t, j)$ are the same except for $A$ extra red dots that appear in the calculation of$S_1$ but not $S_2$ and $B$ extra red dots that appear in the calculation of $S_2$ but not $S_1$. The lemma therefore follows.
\end{proof}

Define $F_0(t)$ to be the number of free slots in the region $[j - t, j)$ before $\overline{O}$ is performed (that is, when the elements $I$ are in the table without any tombstones). Define $\mathbb{V}_t$ to be the indicator random variable for the event that 
$$\overline{c}_j = \overline{\surplus}(I, \overline{O}, [j - t, j)) - F_0(t).$$
One thing that will be critical is that $\mathbb{V}_t = 1$ for the interval $[j - t, j)$ that determines the crossing number $\overline{c}_j$ in Lemma \ref{lem:unorderedcrossingnumbertosurplus}.

With the definition of $\mathbb{V}_t$ in mind, we now state and prove the main technical lemma of the section:

\begin{lemma}
Consider a point in time $\ell$ in operation sequence $\overline{O}$. Let $X_{t, \ell}$ be the number of insertions $x_i$ prior to time $\ell$ that hash to $h(x_i) \in [j - t, j)$, and let $Y_{t, \ell}$ be the number of deletions $y_i$ in $\overline{O}$ that remove a $y_i$ from a position in $[j - t, j)$. 

We have that
$$\overline{\surplus}(I, \overline{O}, [j - t, j)) \cdot \mathbb{V}_t \le 2 \cdot \surplus(O, [j - t, j)) \cdot \mathbb{V}_t + O\left(\max_\ell (Y_{t, \ell} - X_{t, \ell})/\beta\right) + G_t,$$
where $G_t$ is a geometric random variable with mean $O(1)$.
\label{lem:hardgeo}
\end{lemma}
\begin{proof}
Since $t$ is fixed for the duration of this lemma, we will use $F_0$, $\mathbb{V}$, $X_\ell$, and $Y_\ell$ as shorthands for $F_0(t)$, $\mathbb{V}_t$, $X_{t, \ell}$, and $Y_{t, \ell}$.

Let $F_\ell$ be the number of free slots in $[j - t, j)$ at time $\ell$ in $\overline{O}$. Then,
$$F_\ell \le F_0 + Y_\ell - X_\ell + \overline{c}_j.$$
By the definition of $\mathbb{V}$, it follows that
\begin{equation}\mathbb{V} \cdot F_\ell \le Y_\ell - X_\ell + \overline{\surplus}(I, \overline{O}, [j - t, j)).
\label{eq:V0}
\end{equation}

Consider some threshold $s > 0$. Define $\mathbb{S}_{s}$ to be the indicator random variable for the event that the following three conditions hold simultaneously:
\begin{itemize}
    \item $\surplus(O, [j - t, j)) \le s / 2$, 
    \item $\overline{\surplus}(I, \overline{O}, [j - t, j)) = s$, 
    \item and $\max_\ell (Y_\ell - X_\ell) \le \beta s / 8$.
\end{itemize}   

Further define $\mathbb{A}_{\ell, s}$ to be the indicator for the event that $F_\ell \le s + \beta s / 8$. Note that, by \eqref{eq:V0}, whenever $\mathbb{V} \cdot \mathbb{S}_{s} = 1$, we also have $\mathbb{A}_{\ell, s} = 1$ for all $\ell$. Finally, define $\mathbb{K}_\ell$ to be the indicator random variable for the event that the $\ell$-th operation in $\overline{O}$ is a deletion with a failed deletion target $k_\ell \in [j - t, j)$, and define $\mathbb{K} = \sum_\ell \mathbb{K}_\ell$. Then,
since $\mathbb{V} \cdot \mathbb{S}_{s} \le \mathbb{A}_{\ell, s}$ for all $\ell$, we have
\begin{equation} \mathbb{S}_{s} \cdot \mathbb{V} \cdot \mathbb{K} \le \sum_\ell \mathbb{A}_{\ell, s} \mathbb{K}_\ell.\label{eq:Kl}\end{equation}
The latter sum is bounded above by a binomial random variable with mean $R \cdot (s + \beta s / 8) / n = s/\beta + s/8$. Using the fact that $\beta \ge 8$, the binomial random variable has mean at most $s/4$. Thus the probability that it exceeds $s/2$ is at most $2^{-\Omega(s)}$.

On the other hand, by Lemma \ref{lem:phantom}, we have deterministically that $$\overline{\surplus}(I, \overline{O}, [j - t, j)) \le \surplus(O, [j - t, j)) + \mathbb{K}.$$
We also have that, if $\mathbb{S}_s$ occurs, then 
$$\overline{\surplus}(I, \overline{O}, [j - t, j)) - \surplus(O, [j - t, j)) \ge s/2.$$
It follows that, if $\mathbb{S}_s$ occurs, then $\mathbb{K} \ge s/2$. Combining this with \eqref{eq:Kl}, we can conclude that, if $\mathbb{S}_s \cdot \mathbb{V}$ is to occur, then we must have
$$\sum_\ell \mathbb{A}_{\ell, s} \mathbb{K}_\ell \ge s / 2,$$
which we have already concluded occurs with probability $2^{-\Omega(s)}$. Critically, this means that
$$\Pr\left[\mathbb{S}_s \cdot \mathbb{V}\right] \le 2^{-\Omega(s)},$$
which further implies by a union bound that
\begin{equation}\Pr\left[\sum_{s' \ge s} \mathbb{S}_{s'} \cdot \mathbb{V}\right] \le 2^{-\Omega(s)},
\label{eq:sgeo}
\end{equation}

To get from here to the end of the proof, observe that, by the definition of $\mathcal{S}_s$, we have
\begin{align*}  \mathbb{I}_{\overline{\surplus}(I, \overline{O}, [j - t, j)) \ge s} \cdot \mathbb{V} \le & \sum_{s' \ge s} \mathbb{S}_{s'} \cdot \mathbb{V} + \mathbb{I}_{\surplus(O, [j - t, j)) \ge s/2} + \mathbb{I}_{\max_\ell |Y_\ell - X_\ell| \ge  \beta s/8}.\end{align*}
Summing over $s$, it follows that
$$\overline{\surplus}(I, \overline{O}, [j - t, j)) \cdot \mathbb{V} \le \sum_{s}\sum_{s' \ge s} \mathbb{S}_{s'} \cdot \mathbb{V} + 2 \cdot \surplus(O, [j - t, j)) + \frac{8}{\beta} \max_\ell |Y_\ell - X_\ell|.$$
The first sum is, by \eqref{eq:sgeo}, bounded above by a geometric random variable $G$ with mean $O(1)$. 

\end{proof}

In order to make use of Lemma \ref{lem:hardgeo}, we will also need a bound on $O(\max_\ell (Y_{t, \ell} - X_{t, \ell})/\beta)$. 

\begin{lemma}
    Let $t = \Omega(x^2 \log^{1.5} t)$. Out of the first $2\ell$ operations of $\overline{O}$, let $X_\ell$ be the number that insert an element with hash in $[j - t, j)$ and let $Y_\ell$ be the number that delete an element from a position in $[j - t, j)$. Then,
    $$\E[\max(0, \max_\ell (Y_\ell - X_\ell) - t/x] \le O(1 / t^2).$$ 
    \label{lem:XlYl}
\end{lemma}
\begin{proof}
    Each insertion in $\overline{O}$ has at probability exactly $\frac{t}{n}$ of hashing into $[j - t, j)$ and each deletion has probability at most $\frac{t}{(1 - x^{-1})n}$ of selecting an item whose position is in $[j - t, j)$ (since the deletion is of a random element out of $(1 - x^{-1}) n$ options). It follows that
    $$\left\{X'_\ell = X_\ell - \frac{t\ell}{n}\right\}_{\ell}$$
    is a martingale and that
    $$\left\{Y'_\ell = Y_\ell - \frac{t\ell}{(1 - x^{-1}) n}\right\}_{\ell}$$
    is a supermartingale. 

    Now, consider a threshold $\tau$ (to be determined later) and further define $\hat{\ell}$ to be the smallest $\ell$ such that
    $$Y'_{\ell} - X'_{\ell} \ge \tau,$$
    or to be $\hat{\ell} = |O|/2$ if no such $\ell$ exists. Then define sequences $X''_\ell$ and $Y''_\ell$ by
    $$(X''_\ell, Y''_\ell) = \begin{cases}(X'_\ell, Y'_\ell) & \text{ if }\ell \le \hat{\ell} \\ (X'_{\hat{\ell}}, Y'_{\hat{\ell}}) & \text{ otherwise.} \end{cases}$$
    The sequence $\{X''_\ell\}_{\ell = 1}^{\ell = R}$ is a martingale and the sequence $\{Y''_\ell\}_{\ell = 1}^{\ell = R}$ is a supermartingale, so we can apply the multiplicative version of Azuma's inequality \cite{kuszmaul2021multiplicative} to conclude that, for $\mu_1 := \frac{Rt}{n} = t / (2\beta)$, we have
    $$\Pr\left[X''_{R} < - \delta \right] < e^{- \delta^2 / (2 \mu_1)}$$
    and that, for $\mu_2 := \frac{R t}{(1 - x^{-1}) n} = t / (2 \cdot (1 - x^{-1}) \cdot \beta)$, we have
    $$\Pr\left[Y''_{R} > \delta \right] < e^{-\delta^2 / (2\mu + \delta)}.$$
     Since $\mu_1 = \Theta(\mu_2) = \Theta(t / \beta)$, it follows that
     $$\Pr[Y''_{R} - X''_{R} \ge \tau ] \le e^{-\Omega(\tau^2 / (t/\beta + \tau))}.$$
     On the other hand, if 
     $$\max_\ell (X_\ell - Y_\ell) \ge \tau + \frac{t R}{(1 - x^{-1})n} - \frac{t R}{n},$$
     then we will necessarily also have $Y''_\ell - X''_\ell \ge \tau$. It follows that
    $$\Pr\left[\max_\ell (X_\ell - Y_\ell) \ge \tau + \frac{tR}{(1 - x^{-1} )n} - \frac{t R}{n}\right] \le e^{-\Omega(\tau^2 / (t/\beta + \tau))}.$$
    Since 
    $$\frac{tR}{(1 - x^{-1}) n} - \frac{tR}{n} \le \frac{2 tR}{xn} = \frac{2t}{x\beta},$$
    it follows that
    $$\Pr\left[\max_\ell (X_\ell - Y_\ell) \ge \tau + \frac{2t}{x\beta} \right] \le e^{-\Omega(\tau^2 / (t/\beta + \tau))}.$$
    
    This, in turn, implies that, for a sufficiently large positive constant $\gamma$, we have
    $$\E\left[\max\left(0, \max_\ell (X_\ell - Y_\ell) - \frac{2t}{x\beta} - \gamma \sqrt{t/\beta} \sqrt{\log (t/\beta)}\right)\right ] \le 1 / \poly(t/ \beta) \le O(1 / t^2).$$
    Since, by assumption, $\frac{t}{x\beta} + \gamma \sqrt{t/\beta} \sqrt{\log (t/\beta)}) \le O(t / x)$, this completes the proof.
\end{proof}







We can now prove the upper-bound direction of Proposition \ref{prop:translate}:

\begin{lemma}
We have
  $$\E\left[\max_t \left(\overline{\surplus}(I, \overline{O}, [j - t, j)) - F_0(t)\right)\right] \le O\left(\frac{R}{n} x \log^{1.5} x + x\right).$$
  \label{lem:unorderedupper}
\end{lemma}
\begin{proof}
Define the $\maxx(\mathcal{S})$ operator to be the $\max$ operation modified to return $0$ if it would have otherwise returned a negative number. As a shorthand for this proof, define $\overline{\surplus}(t) = \overline{\surplus}(I, \overline{O}, [j - t, j))$ and $\surplus(t) = \overline{\surplus}(O, [j - t, j))$. Finally, define $\tau = x^2 \log^{1.5} x$. Then, 
\begin{align*}
    & \phantom{=} \E[\max_{t \ge 0} (\overline{\surplus}(t) - F_0(t))] 
    \\ &= \E[\maxx_{t \ge 0} (\overline{\surplus}(t) - F_0(t)) \cdot \mathbb{V}_t] \tag{by Lemma \ref{lem:unorderedcrossingnumbertosurplus}}\\
    & \le  \E[\maxx_{t \ge \tau} (\overline{\surplus}(t) - F_0(t)) \cdot \mathbb{V}_t] + \E[\maxx_{t < \tau} (\overline{\surplus}(t) - F_0(t)) \cdot \mathbb{V}_t]. 
\end{align*}    
By Lemma \ref{lem:phantom}, for $t < \tau$, $\overline{\surplus}(t) - \surplus(t)$ is at most the number of failed deletions with targets in $[s - \tau, s)$. The expected number of such deletions is at most $R \cdot \frac{1}{x} \cdot \frac{\tau}{n} = \tau / \beta$, since each of $R$ deletions has probability $1/x$ of failing and, conditioned on failing, probability $\tau / n$ of having a failed target in $[s - \tau, s)$. We therefore have that
$$\E[\maxx_{t < \tau} (\overline{\surplus}(t) - F_0(t)) \cdot \mathbb{V}_t] \le O\left(\frac{\tau}{x \beta}\right) + \E[\maxx_{t < \tau} (\surplus(t) - F_0(t)) \cdot \mathbb{V}_t], $$
which we know from the analysis in Lemma \ref{lem:crossingupper} is at most
$$ O\left(\frac{\tau}{x \beta}\right) + O\left(\frac{R}{n} x \log^{1.5} x\right).$$

Focusing on $t \ge \tau = x^2 \log^{1.5} x$, we have by Lemma \ref{lem:hardgeo} that
\begin{align*}
    & \phantom{=} \E[\maxx_{t \ge \tau} (\overline{\surplus}(t) - F_0(t)) \cdot \mathbb{V}_t] \\
    & \le  \E\left[\maxx_{t \ge \tau} \left(2\cdot \surplus(t) + O\left(\max_{\ell} (Y_{t, \ell} - X_{t, \ell})/\beta\right) + G_t - F_0(t)\right)\right],
    \end{align*}  
where $G_t$ is bounded above by a geometric random variable with mean $O(1)$. This, in turn, is at most
\begin{align*}
&    \sum_{t \ge \tau} O(1/t^2) +  \E\left[\maxx_{t \ge \tau} \left(2 \cdot \surplus(t) + O\left(\max_{\ell} (Y_{t, \ell} - X_{t, \ell})/\beta\right) + O(\log t) - F_0(t)\right)\right] \\
&  O(1)  + \E\left[\maxx_{t \ge \tau} \left(2 \cdot \surplus(t) + O\left(\max_{\ell} (Y_{t, \ell} - X_{t, \ell})/\beta\right) + O(\log t)-  F_0(t)\right)\right].
\end{align*}
Applying Lemma \ref{lem:XlYl}, this is at most
\begin{align*}
& \phantom{\le} O(1)  + \sum_{t \ge \tau} O(1/t^2) + \E\left[\maxx_{t \ge \tau} \left(2 \cdot \surplus(t) + O\left(\frac{t}{x\beta}\right) + O(\log t)- F_0(t)\right)\right]  \\
& \le O(1) + \E\left[\maxx_{t \ge \tau} \left(2 \cdot \surplus(t) + O\left(\frac{t}{x\beta}\right) + O(\log t)- F_0(t)\right)\right]  \\
& \le O(1) + \E\left[\maxx_{t \ge \tau} \left(2 \cdot \surplus(t) + \frac{t}{8x} + O(\log t)- F_0(t)\right)\right] \tag{since $\beta$ is at least a large constant} \\
    & \le O(1) + \E\left[\maxx_{t \ge \tau} \left(2\cdot \surplus(t) - \frac{t}{4x} + \frac{t}{8 x}\right)\right] + \E\left[\maxx_{t \ge \tau} \left(\frac{t}{4x} - F_0(t)\right)\right]. \\
        & \le O(1) + 2 \cdot \E\left[\maxx_{t \ge \tau}  \surplus(t) - \frac{t}{16x}\right] + \E\left[\maxx_{t \ge \tau} \frac{t}{4x} - F_0(t)\right]. \\
\end{align*}    
    We can now continue with the analysis as in Lemma \ref{lem:crossingupper} to bound this entire quantity by
    $$ O\left(1\right) + O\left(\frac{R}{n} x \log^{1.5} x + x\right).$$
    Combining our bounds for the cases of $t \ge \tau$ and $t < \tau$, we have that
    \begin{align*} \E[\maxx_{t \ge t} (\overline{\surplus}(t) - F_0(t))] & \le O\left(\frac{\tau}{x \beta} + nx\right) + O\left(\frac{R}{n} x \log^{1.5} x + x\right) \\
    & \le O\left(\frac{x^2 \log^{1.5} x}{x \beta}\right) + O\left(\frac{R}{n} x \log^{1.5} x + x\right)  \\
    & \le O\left((x/\beta) \log^{1.5} x \right) + O\left(\frac{R}{n} x \log^{1.5} x + x\right)  \\
    & \le O\left(\frac{R}{n} x \log^{1.5} x + x\right). \\
    \end{align*}
\end{proof}

Finally, we conclude the section by proving lower-bound side of Proposition \ref{prop:translate} in two lemmas. The first establishes a lower bound of $\Omega\left(\frac{R}{n} x \log^{1.5} x \right)$, and is a relatively immediate application of Lemma \ref{lem:phantom}:

\begin{lemma}
    We have
  $$\E\left[\max_t \left(\overline{\surplus}(I, \overline{O}, [j - t, j)) - F_0(t)\right)\right] \ge \Omega\left(\frac{R}{n} x \log^{1.5} x \right).$$
  \label{lem:unorderedlower1}
\end{lemma}
\begin{proof}
Let $q$ be a sufficiently small positive constant and set $t = q \cdot \frac{x^2}{\beta} \log^{1.5} x$. Then, by the analysis in Lemma \ref{lem:crossinga} (which uses Corollary \ref{cor:surpluslower-c}), we have
\begin{align*}\E[\surplus(O, [j - t_0, j)) - F_0(t_0)] & \ge  \Omega\left(\frac{R}{n} x \log^{1.5} x\right). 
\end{align*}
By Lemma \ref{lem:phantom}, it follows that
\begin{align*}\E[\surplus(I, \overline{O}, [j - t_0, j)) - F_0(t_0)] & \ge  \Omega\left(\frac{R}{n} x \log^{1.5} x \right) - \E[B], 
\end{align*}
where $B$ is the number of failed deletions in $\overline{O}$ with replacement targets in $[j -t, j)$. In expectation there are $O(R/x)$ failed deletions, each of which has probability $O(t_0 / n)$ of having a replacement target in $[j - t, j)$. Therefore, 
$$\E[B] \le O\left(\frac{R t_0}{x n}\right) = O\left(\frac{R}{n} \cdot \frac{x}{\beta} \log^{1.5} x\right).$$
Using the fact that $\beta$ is at least a sufficiently large positive constant, the lemma follows.
\end{proof}

To complete the lower-bound side of Proposition \ref{prop:translate}, we also need to establish a lower bound of $\Omega(x)$. This turns out to be a bit more tricky, since we cannot rely on the workload that was used in Section \ref{sec:insertionsurplus} to prove the analogous lower bound for ordered linear probing. Nonetheless, with an alternative path, we can still get our desired result:
\begin{lemma}
    We have
  $$\E\left[\max_t \left(\overline{\surplus}(I, \overline{O}, [j - t, j)) - F_0(t)\right)\right] \ge \Omega\left(x\right).$$
  \label{lem:unorderedlower2}
\end{lemma}
\begin{proof}
If $\beta \le \log^{1.5} x$, then the result follows from Lemma \ref{lem:unorderedlower1}. Suppose for the rest of the proof that $\beta \ge \log^{1.5} x$, and consider $t = x^2 / \gamma$ for some sufficiently large positive constant $\gamma$.

We will prove three facts:
\begin{enumerate}
\item that, with probability $\Omega(1)$, $F_0(t) = 0$;
\item that $\E[\surplus(O, [j - t_0, j))] \ge \Omega(x)$;
\item and that $\E[B \mid F_0(t) = 0] = o(x)$.
\end{enumerate}

Before proving these facts, however, let us observe how they would allow us to complete the proof of the lemma.

Combining the first two facts with the observation that $F_0(t)$ is independent of $\surplus(O, [j - t_0, j))$ (because the generation process for the latter is oblivious to the initial set $I$ of elements), we get that
$$\E[\surplus(O, [j - t_0, j)) \cdot \mathbb{I}_{F_0(t) = 0}] \ge \Omega(x),$$
and thus that
$$\E[(\surplus(O, [j - t_0, j)) - F_0(t)) \cdot \mathbb{I}_{F_0(t) = 0}] \ge \Omega(x).$$
Combining this with the third fact, it follow that
$$\E[(\surplus(O, [j - t_0, j)) - B - F_0(t)) \cdot \mathbb{I}_{F_0(t) = 0}] \ge \Omega(x),$$
which, by Lemma \ref{lem:phantom}, implies that
$$\E[(\overline{\surplus}(I, O, [j - t_0, j)) - F_0(t)) \cdot \mathbb{I}_{F_0(t) = 0}] \ge \Omega(x).$$
Thus, if we can prove the three itemized facts, then the proof of the lemma will be complete.

The first fact follows, albeit slightly indirectly, from several other standard facts. By the standard analysis of linear probing \cite{Knuth63}, we know that, at the beginning of the rebuild window, the expected length of the run containing a position $i - t$ is $\Theta(x^2)$. We also know that the probability of the run having length $k \cdot x^2$ drops off at least exponentially fast in $k$ (see, e.g., Proposition 1 of \cite{bender2022linearfull}). The only way that both these statements can be true is if the run starting at position $i - t$ has probability $\Omega(1)$ of having length $\Omega(x^2)$. Since $t = x^2 / \gamma$ for a sufficiently large positive constant $\gamma$, it follows that, with probability $\Omega(1)$, there are no free slots in $[j - t, j)$.

The second fact follows from Corollary \ref{cor:surpluslower-c}), which tells us that
\begin{align*}\E[\surplus(O, [j - t_0, j))] & = \Omega\left(\sqrt{t / \beta} \log^{0.75} x\right) = \Omega\left(\sqrt{x^2 / \beta} \log^{0.75} x\right) = \Omega(x). 
\end{align*}

To prove the third fact, notice that, even if we condition on $F_0(t) = 0$, then we can still conduct the following analysis on $B$. With high probability in $n$, the total number of failed deletions is $O(R / x)$ and each such failed deletion independently (and regardless of the fact that $F_0(t) = 0$) has probability at most $O(t / n) = O(x^2 / n)$ of using a replacement target whose position is in $[j - t, j)$. It follows that, even conditioned on $F_0(t)$, we have $\E[B] \le O((R / x) \cdot (x^2 / n)) = O(Rx / n) = o(x)$.

\end{proof}

Combined, Lemmas \ref{lem:unorderedupper}, \ref{lem:unorderedlower1}, and \ref{lem:unorderedlower2} prove Proposition \ref{sec:finalproposition}.

\section{Acknowledgements}

Mark Braverman is supported in part by the NSF Alan T. Waterman Award, Grant No. 1933331. William Kuszmaul was partially supported by a Harvard Rabin Postdoctoral Fellowship and by a Harvard FODSI fellowship under NSF grant DMS-2023528.

\bibliographystyle{plainurl}
\bibliography{linearprobe}
\end{document}